\documentclass[11pt]{article}
\usepackage{graphicx} 

\usepackage{style}
\usepackage{authblk}

\newcommand{\blind}{0}

\title{Testing hypotheses via orthogonalization}
\if\blind1 {\author{\vspace{-6ex}}\date{\vspace{-6ex}}
} \else{
\author[1]{Ameer Dharamshi}
\author[1]{Runjia Zou}
\author[1,2]{ and Daniela Witten}
\affil[1]{Department of Biostatistics, University of Washington}
\affil[2]{Department of Statistics, University of Washington}
}\fi

\begin{document}

\maketitle

\begin{abstract} 
Classical hypothesis testing frameworks break down in contemporary settings in which null hypotheses are increasingly abstract, the same data are used to both generate and test hypotheses, and minimal assumptions about the underlying data are made. In this work, we propose a new framework for conducting valid hypothesis tests in broad contexts. We propose to add and subtract external noise generated from a symmetric shift-family to our data, $X$, to partition it into two pieces, $\Xt{1}$ and $\Xt{2}$. We provide a generic strategy for orthogonalizing $\Xt{2}$ against $\Xt{1}$ under the null hypothesis $H_0$, then show that testing whether the orthogonalization was successful provides a valid test of $H_0$ under mild assumptions. 
Remarkably, this framework extends naturally to the post-selection inference setting: we simply select a hypothesis on $\Xt{1}$, then perform orthogonalization under the selected null. As our approach neither requires pre-specification of the selection mechanism, nor is  restricted to a small class of data-generating distributions, it dramatically expands the settings for which valid post-selection inference can be conducted. We showcase the flexibility of our proposal in several case studies involving challenging pre-specified null hypotheses and post-selection inference scenarios. 

\end{abstract}

\section{Introduction}
\label{sec:intro}

Hypothesis testing has a rich history in the literature of statistics, dating back to foundational work by \cite{fisher1970statistical, fisher1966design, fisher1956statistical}, \cite{wald1950statistical}, \cite{neyman1952lectures}, \cite{neyman2023joint}, and many others; see \cite{lehmann2005testing} for a comprehensive overview of classical hypothesis testing. 
In this paper, we present an entirely different approach to hypothesis testing: we make use of recent ideas from post-selection inference to provide a very general hypothesis testing framework that can be applied to either pre-specified or adaptively-selected hypotheses.

Formally, suppose that we observe $n$ independent random variables $X_i\sim F(\eta_i(\theta^*))$  for $i=1,\dots,n$, where $F: H \rightarrow \mathcal{X}$ is some distributional family parameterized by $\eta_i$, which is itself parameterized by a fixed unknown parameter $\theta^*\in\Theta$. The dimension of $\theta^*$ does not grow with $n$, but  is unrestricted: it could even be the entire distribution function of $X_i$.

Suppose that our goal is to test the hypothesis $H_0:\theta^*\in\Theta_0$ where $\Theta_0\subset\Theta$; for now, $H_0$ is pre-specified. We begin by decomposing each $X_i$ into $\Xt{1}_i$ and $\Xt{2}_i$ by adding and subtracting a particular kind of independent external randomness. Specifically, we let $W_i$ denote noise generated from a \emph{symmetric shift-family centred at zero}, and let $\Xt{1}_i=X_i+W_i$ and $\Xt{2}_i=X_i-W_i$. In this paper, we will provide a generic strategy for \emph{orthogonalizing} $\Xt{2}_i$ against $\Xt{1}_i$ under $H_0$. We will then show that to test $H_0$ it is sufficient to test whether the orthogonalization was successful. This can be accomplished by testing whether a particular moment is zero.

The key to our proposal involves the aforementioned orthogonalization step: our construction of $W_i$ leads to a general expression for the conditional mean of $\Xt{2}_i$ given $\Xt{1}_i$ that involves a ratio of two expectations, each of which can be computed with either a simple Monte Carlo simulation or, in some cases, with  sample means. Consequently, this  orthogonalization strategy --- and hence, our new approach to testing --- is broadly applicable in a wide range of settings, and specifically leads to a very straightforward test in cases where alternative approaches require careful derivation or analysis. For example, we show in Section \ref{subsec:twosample} that orthogonalization allows us to immediately derive a nonparametric test for a difference in distribution between two samples, a topic of recent interest for which far more involved procedures have also been proposed \citep{hore2026distribution}. 

It turns out that our orthogonalization approach is deeply related to recent work in the selective inference literature, and in particular, to the ideas of \emph{data thinning} \citep{neufeld2023data, dharamshi2023generalized} and \emph{data fission} \citep{leiner2022data}. Data thinning seeks to decompose a random variable $X_i$ into two independent components $\Xt{1}_i$ and $\Xt{2}_i$ so that a hypothesis can be selected on the basis of $\Xt{1}_i$ and tested using $\Xt{2}_i$. However, independence between $\Xt{1}_i$ and $\Xt{2}_i$ requires stringent distributional assumptions. 
Data fission instead seeks to decompose $X_i$ into dependent components $\Xt{1}_i$ and $\Xt{2}_i$. As before, $\Xt{1}_i$ is used to select a hypothesis, but now inference must be conducted using the conditional distribution of $\Xt{2}_i$ given $\Xt{1}_i$. However, data fission is also restricted to parametric settings, and moreover in practice, the inference step is not tractable outside of a small number of special cases \citep{dharamshi2024decomposing, neufeld2024discussiondatafissionsplitting}.  
In this paper, we also construct $\Xt{1}_i$ and $\Xt{2}_i$ that are dependent, but in such a way that $\Xt{2}_i$ can be orthogonalized with respect to $\Xt{1}_i$, \emph{regardless of the distribution of $X_i$}.

Remarkably, it turns out that orthogonalization is actually sufficient for hypothesis testing in settings with either pre-specifed \emph{or} data-driven hypotheses. Specifically, suppose that we decompose $X_i$ into $\Xt{1}_i$ and $\Xt{2}_i$ using symmetric shift-family noise $W_i$, and  select a null hypothesis $H_0(\xt{1})$ that is a function of the realized data $\xt{1}$. Then, to test $H_0(\xt{1})$, it is again sufficient to orthogonalize $\Xt{2}_i$ with respect to $\Xt{1}_i$ \emph{under the selected null $H_0(\xt{1})$}, and then to test if a simple moment condition is equal to zero. Furthermore, \emph{it turns out that this strategy enables valid tests of data-driven hypotheses in settings where no solution was previously available}. For instance, in Section \ref{subsec:clustering}, we will show that we can test for a difference in distribution between subgroups of data identified using clustering algorithms. As pointed out by \cite{gao2020selective}, classical tests fail in this setting as they do not account for the fact that the subgroups are selected from the data; they also show that sample splitting \citep{cox1975note} is  not a viable solution. Existing strategies for conducting inference after clustering are only available if the data are Gaussian \citep{gao2020selective, chen2022selective, yun2023selective}, or belong to a distributional family that admits a tractable data thinning or data fission procedure. As we will see, our proposal provides a pathway towards inference after clustering in far more complex data scenarios.

The remainder of this paper is organized as follows. In Section \ref{sec:framing}, we  describe the process by which an arbitrary null hypothesis can be converted into a relatively simple orthogonality moment condition. Section \ref{sec:cmean} offers a test of this moment. Section \ref{sec:sel} extends our proposal to the post-selection inference setting. Sections \ref{sec:cases} and \ref{sec:data} illustrate our proposal in several simulated case studies and in an application to single-cell RNA sequencing data, respectively. We then conclude with a discussion in Section \ref{sec:discussion}. Proofs of technical results are deferred to the supplementary materials.

\section{Recasting $H_0$ as a test of a moment condition}
\label{sec:framing}

Suppose that we observe $n$ independent random variables $X_i\sim F(\eta_i(\theta^*))$  for $i=1,\dots,n$. We assume that $F: H \rightarrow \mathcal{X}$ is some $p$-dimensional continuous or count-valued distribution with finite second moment parameterized by $\eta_i(\theta^*)$, which in turn is parameterized by a fixed unknown $\theta^*\in\Theta$. Here $\eta_i: \Theta\rightarrow H$ is a known mapping that subsets or transforms $\theta^*$ to allow for heterogeneity among $X_1,\ldots,X_n$; for instance, if the observations belong to two distinct populations, then $\theta^*$ may represent the parameters for both populations, and $\eta_i$ subsets $\theta^*$ to the parameters of the population to which the  $i$th observation belongs. Aside from the second moment condition, we place no restrictions on $F$ or on the dimensionality of $\theta^*$; that is, $F$ could be a parametric family with parameter $\theta^*$, or a nonparametric distribution whose distribution function is $\theta^*$. We denote the probability density/mass function of $F$ as $f_X$, and assume that it is known up to the unknown parameter $\theta^*$. This is not a restrictive assumption, as in the nonparametric case, $\theta^*$ fully characterizes the distribution. Lastly, we write expectations taken with respect to $F(\eta_i(\theta^*))$ as $\E_{\eta_i(\theta^*)}[\cdot]$.
To make this notation concrete, we illustrate it in a parametric and nonparametric setting in Examples \ref{ex:zip} and \ref{ex:np},  respectively; we will re-visit these examples in Section \ref{sec:cases}.

\begin{exmp}
\label{ex:zip}
Suppose that $X_i\overset{iid}\sim \text{ZIP}(\lambda^*,\pi^*)$ where $\text{ZIP}(\lambda,\pi)$ indicates the zero-inflated Poisson distribution with unknown rate $\lambda$ and unknown zero-inflation parameter $\pi$. Using the notation $X_i\sim F(\eta_i(\theta^*))$, $F$ refers to the family of zero-inflated Poisson distributions with both parameters unknown, $\theta^*=(\lambda^*,\pi^*)$, and $\eta_i(\theta^*)=\theta^*$ (i.e., $\eta_i(\cdot)$ is the identity function since the data are independent and identically distributed).
\end{exmp}

\begin{exmp}
\label{ex:np}
Consider a nonparametric two-sample problem in which $X_i\overset{iid}\sim P$ for $i\le n/2$ and $X_i\overset{iid}\sim Q$ for $i>n/2$. Here $P$ and $Q$ are $p$-dimensional continuous distributions with finite second moment. Using the notation $X_i\sim F(\eta_i(\theta^*))$, $F$ refers to the family of all $p$-dimensional continuous distributions with finite second moment, $\theta^*=(\theta_P,\theta_Q)$ where $\theta_P$ and $\theta_Q$ are the distribution functions of $P$ and $Q$, respectively, and $\eta_i(\theta^*)$ returns $\theta_P$ when $i\le n/2$ and $\theta_Q$ otherwise.
\end{exmp}

Our goal is to test some hypothesis about $\theta^*$, which we write as $H_0:\theta^*\in\Theta_0$ where $\Theta_0 \subset \Theta$. For now, we assume that $H_0$ is pre-specified; we return to the task of testing data-driven hypotheses in Section \ref{sec:sel}. As discussed in the introduction, for an arbitrary $H_0$, identifying a suitable test statistic with a convenient null distribution may be a laborious task requiring careful theoretical derivation. In this paper, we take a different approach: rather than hoping that there exists some convenient function of the data that can be used to test $H_0$, we \emph{create} an orthogonality structure in the data that can be used to test an arbitrary $H_0$. Specifically, our strategy is to decompose $X$ into two folds, $\Xt{1}$ and $\Xt{2}$, in such a way that when $H_0$ is true, $\Xt{2}$ can be orthogonalized with respect to $\Xt{1}$. Then, to test $H_0$, it is sufficient to test whether we have successfully  orthogonalized $\Xt{2}$ with respect to $\Xt{1}$.

In the remainder of this section, we discuss in detail the process of reformulating $H_0$ into a test of orthogonality. We begin with Algorithm \ref{alg:split}, which injects a particular kind of external noise into a realized dataset $x$ to produce two folds, $\xt{1}$ and $\xt{2}$. We emphasize that the user specifies the $p$-dimensional symmetric shift-family distribution $R(\phi, \Sigma)$ along with its variance $\Sigma$. 

\begin{algorithm}
\label{alg:split}
\textcolor{white}{.}

Input: Observed data $x=(x_1,\dots,x_n)$ drawn from $X_i\overset{ind}{\sim} F(\eta_i(\theta^*))$, and a user-specified $p$-dimensional symmetric shift-family distribution $R(\phi,\Sigma)$ with mean $\phi$, user-specified variance $\Sigma$, and density $f_R(\cdot;\phi,\Sigma)$. 
\begin{enumerate}
    \item For $i=1,\dots,n$, generate one realization $w_i$ from $W_i\overset{iid}{\sim} R(0,\Sigma)$ independently of $x$.
    \item Define $\Xt{1}_i=X_i+W_i$ and $\Xt{2}_i=X_i-W_i$, then compute $\xt{1}_i=x_i+w_i$ and $\xt{2}_i=x_i-w_i$.
    \item Return $\xt{1}=(\xt{1}_1,\dots,\xt{1}_n)$ and $\xt{2}=(\xt{2}_1,\dots,\xt{2}_n)$.
\end{enumerate}
\end{algorithm}

\begin{remark}
\label{rem:Rdist}
The ideas in this section in principle hold for any choice of $R$ with finite second moment. However, as we will see in Section \ref{sec:cmean}, the use of symmetric shift-family noise in Algorithm \ref{alg:split} will lead to useful structure in the joint distribution of $(\Xt{1}_i,\Xt{2}_i)$. If $X$ is continuous, we suggest choosing $R$ to be Gaussian and $\Sigma$ to be diagonal. If $X$ is count-valued, we suggest that each coordinate of $R$ should follow either a Skellam distribution with equal parameters (i.e., the difference between two independent and identically distributed Poisson random variables) or a discrete uniform distribution. 
The value of $\Sigma$ is user-specified, and all of our results hold for any choice of $\Sigma$, though its value has implications for power; see Section \ref{subsec:power}.
\end{remark}

\begin{remark}
Algorithm \ref{alg:split} resembles randomization strategies common in the post-selection inference literature; e.g., \cite{rasines2021splitting}. We discuss this connection in Section \ref{sec:sel}.
\end{remark}

Next, recall that the conditional mean of a random variable $A$ given a random variable $B$ is defined as the projection of $A$ onto the space of measurable functions of $B$ \citep{van2000asymptotic}. That is, $\E[A|B]=\arg\min_{g}\E[(A-g(B))^2]$. Proposition \ref{prop:orthogonality} states the well-known orthogonality property of conditional means.

\begin{proposition}
\label{prop:orthogonality}
Let $\E[A|B]$ denote the conditional mean of $A$ given $B$. For every function $g$,
\begin{equation}
\label{eq:orth}
\E\left[(A-\E[A|B])g(B)\right]=0.
\end{equation}
\end{proposition}

Why are Algorithm \ref{alg:split} and Proposition \ref{prop:orthogonality} useful?
Consider any function $h(\xt{1}_i,\Theta_0)$ of the following form:
\begin{equation}
\label{eq:hcases}
h(\xt{1}_i,\Theta_0)=\begin{cases}
    \E_{\eta_i(\theta^*)}[\Xt{2}_i|\Xt{1}_i=\xt{1}_i] &\text{if }\,\theta^*\in\Theta_0\text{, i.e., if }H_0\text{ holds}, \\
    k(\xt{1}_i)&\text{otherwise}
\end{cases}
\end{equation}
where $\E_{\eta_i(\theta^*)}[\Xt{2}_i|\Xt{1}_i=\xt{1}_i]=\E_{\eta_i(\theta^*)}[X_i-W_i|X_i+W_i=\xt{1}_i]$ is taken with respect to the conditional distribution implied by $X_i\sim F(\eta_i(\theta^*))$ and $W_i\sim R(0,\Sigma)$, and $k(\xt{1}_i)$ is some function designed to \emph{not} equal $\E_{\eta_i(\theta^*)}[\Xt{2}_i|\Xt{1}_i=\xt{1}_i]$ when $\theta^*\not\in\Theta_0$.

If $H_0$ is true, then $h(\xt{1}_i,\Theta_0)$ \emph{must satisfy the orthogonality relationship outlined in Proposition \ref{prop:orthogonality}}. This fact motivates the following orthogonality hypothesis: 
\begin{equation}
\label{eq:H0prime}
H_0'(g): \E_{\eta_i(\theta^*)}\left[\left(\Xt{2}_i-h\left(\Xt{1}_i,\Theta_0\right)\right)g(\Xt{1}_i)^\top\right]=0 \quad \forall i=1,\dots,n,
\end{equation}
where $g$ is a real $q$-dimensional test function to be specified by the user; that is, $g:\mathcal{X}^{(1)}\rightarrow\mathbb{R}^q$ where $\mathcal{X}^{(1)}$ is the sample space of $\Xt{1}$ (and $\Xt{2}$). The null hypothesis $H_0'(g)$ states that $h$ successfully orthogonalizes $\Xt{2}$ against $\Xt{1}$. Then, since $H_0$ implies $H_0'(g)$ by construction, it follows from a contrapositive argument that we should reject $H_0$ when we reject $H_0'(g)$. That is, \emph{we can test $H_0$ by testing $H_0'(g)$}. Moreover, as we will see in Section \ref{sec:cmean}, for $\Xt{1}$ and $\Xt{2}$ constructed using Algorithm \ref{alg:split}, $h(\xt{1}_i,\Theta_0)$ can often be reliably estimated, leading to a straightforward test for $H_0'(g)$. 

\begin{remark}
If $X_i$ are independent and identically distributed under $H_0$, \eqref{eq:H0prime} simplifies to 
$$
H_0'(g): \E\left[\left(\Xt{2}-h\left(\Xt{1},\Theta_0\right)\right)g(\Xt{1})^\top\right]=0.
$$
\end{remark}

Algorithm \ref{alg:test} consolidates these ideas into a testing procedure  (see also  Figure \ref{fig:schematic}), and  Theorem \ref{thm:t1e} establishes its Type I error control.

\begin{algorithm}[Testing hypotheses via orthogonalization]
\label{alg:test}
\textcolor{white}{.}

Input: Observed data $x_i$ drawn from $X_i\overset{ind}{\sim} F(\eta_i(\theta^*))$ for $i=1,\dots,n$;  a user-specified $p$-dimensional symmetric shift-family distribution $R(\phi,\Sigma)$ with mean $\phi$, user-specified variance $\Sigma$, and density $f_R(\cdot;\phi,\Sigma)$; a null hypothesis $H_0:\theta^*\in\Theta_0$ where $\Theta_0 \subset \Theta$;   and a test function $g:\mathcal{X}^{(1)}\rightarrow\mathbb{R}^q$. 
\begin{enumerate}
    \item Construct $\xt{1}_i$ and $\xt{2}_i$ using Algorithm \ref{alg:split} with inputs $x_1,\dots,x_n$ and $R(\phi,\Sigma)$. 
    \item Test the hypothesis $H_0'(
    g)$ defined in \eqref{eq:H0prime}. 
\end{enumerate}
\end{algorithm}

\begin{theorem}[Validity of Algorithm \ref{alg:test}]
\label{thm:t1e}
In the context of Algorithm \ref{alg:test}, suppose that $\vartheta_\alpha(\Xt{1},\Xt{2})$ is a valid test for $H_0'(g)$ in \eqref{eq:H0prime} at significance level $\alpha$, in the sense that if $H_0'(g)$ is true, then for any $\alpha\in(0,1)$,
$$
P(\vartheta_\alpha(\Xt{1},\Xt{2}) = 1) \le \alpha.
$$
Then $\vartheta_\alpha(\Xt{1},\Xt{2})$ is also a valid test of $H_0$.
\end{theorem}

Type I error control of $H_0$ in Theorem \ref{thm:t1e} stems from the fact that when $H_0$ holds, $H_0'(g)$ must also hold. This is, however, not an ``if and only if" statement. If $H_0$ fails, then $H_0'(g)$ may not necessarily fail, impacting power. Power arises from the choices of $g$ and $\Sigma$; we discuss this in Section \ref{subsec:power}.

\begin{figure}[h]
\centering
\includegraphics[width=.9\linewidth]{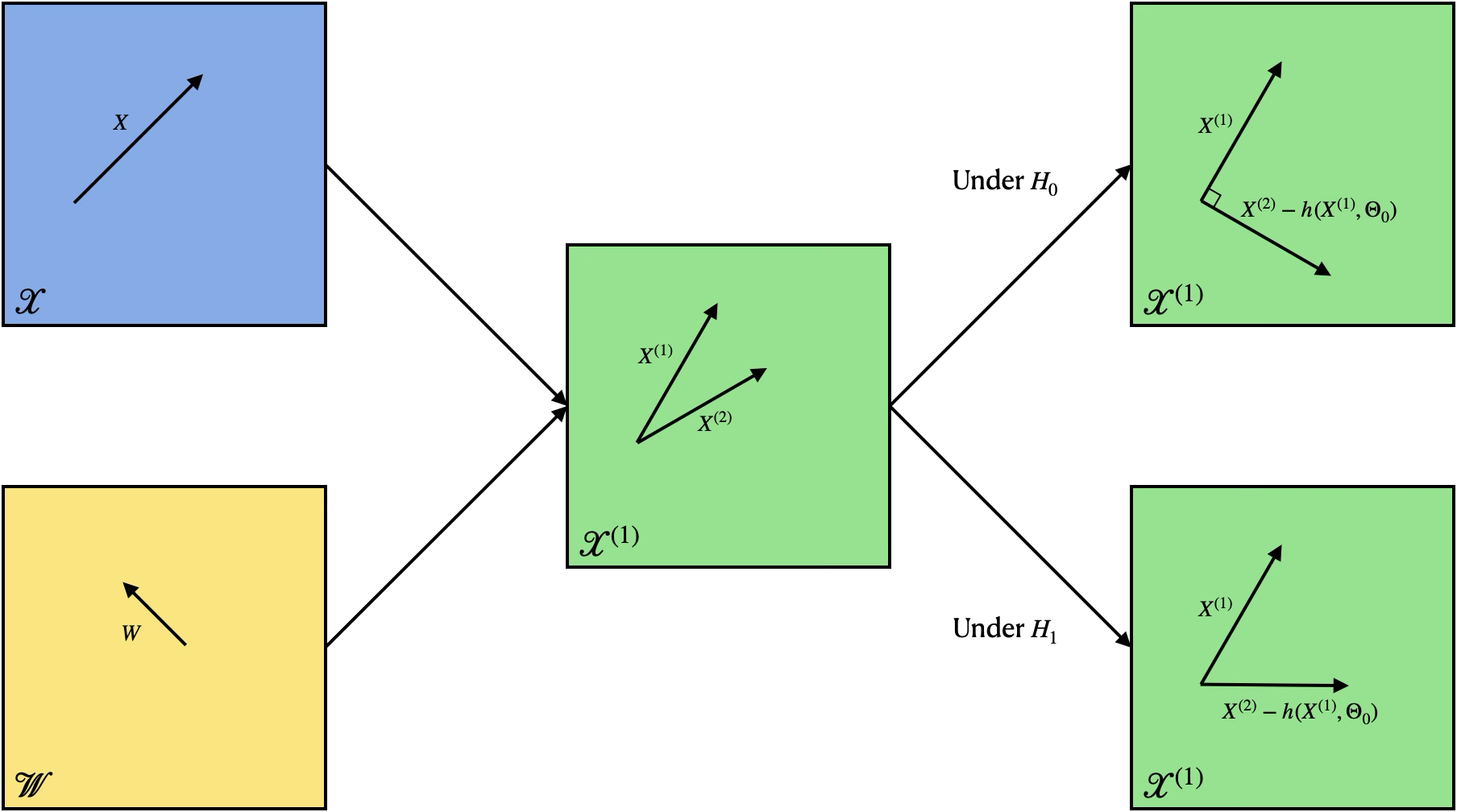}
\caption{A visual depiction of the procedure described in Section \ref{sec:framing}. We begin with data $X$ (upper-left) and generate independent noise $W$ (bottom-left). We then form $\Xt{1}$ and $\Xt{2}$ by adding and subtracting $W$ from $X$ (centre). Finally, if $H_0$ is true, then subtracting $h(\Xt{1},\Theta_0)$ in \eqref{eq:hcases} from $\Xt{2}$ will orthogonalize $\Xt{2}$ with respect to $\Xt{1}$ (upper-right), whereas if $H_0$ is false, subtracting $h(\Xt{1},\Theta_0)$ from $\Xt{2}$ will fail to orthogonalize $\Xt{2}$ with respect to $\Xt{1}$ (bottom-right).}
\label{fig:schematic}
\end{figure}

\section{Testing the orthogonality hypothesis $H_0'(g)$}
\label{sec:cmean}

The cornerstone of our proposal is Step 2 of Algorithm~\ref{alg:test}, in which we convert the potentially challenging problem of testing a null hypothesis $H_0$ into a test of the orthogonality moment condition $H_0'(g)$ in \eqref{eq:H0prime}.
In this section, we develop a two-stage procedure for testing $H_0'(g)$. Our key insight is that testing $H_0'(g)$ amounts to testing whether a population moment equals zero in the presence of an unknown nuisance function. We first construct an estimator $\hat h_n(\xt{1}_i,\Theta_0)$ of $h(\xt{1}_i,\Theta_0)$ by exploiting properties of the joint distribution of $(\Xt{1}_i,\Xt{2}_i)$ induced by the use of symmetric shift-family noise in Algorithm \ref{alg:split}. We then design a test statistic using the sample analog of \eqref{eq:H0prime} in which $\hat h_n(\xt{1}_i,\Theta_0)$ replaces $h(\xt{1}_i,\Theta_0)$. Finally, we derive the asymptotic distribution of our test statistic using ideas from semiparametric theory to account for estimation uncertainty stemming from $\hat h_n(\xt{1}_i,\Theta_0)$.

\begin{remark}
Testing moment conditions such as \eqref{eq:H0prime} in the presence of nuisance parameters/functions is a task found throughout statistics; two examples include the classical J-tests for over-identification in the context of instrumental variables \citep{hansen1982large} as well as more recent work on conditional independence testing \citep{shah2020hardness}. Our approach to testing $H_0'(g)$ takes inspiration from this rich literature, though it differs in important respects. Most fundamentally, our target, $H_0$, is not necessarily a moment condition natively; rather, our proposal manufactures the moment condition in $H_0'(g)$ as a means of easing the burden of testing $H_0$. As a consequence of the fact that we are building $H_0'(g)$ using user-selected $R$ and $g$, thoughtful choices will yield considerable simplifications in the theoretical analysis that follows.
\end{remark}

\subsection{Specifying an estimator for $h(\xt{1}_i,\Theta_0)$}

In general, an analyst will not have \emph{a priori} access to a function $h(\xt{1}_i,\Theta_0)$ satisfying \eqref{eq:hcases}; rather, it must be estimated. This is a nontrivial task, as  $h(\xt{1}_i,\Theta_0)$ serves two contrasting purposes: 
\begin{enumerate}
    \item[\emph{P1:}] \emph{When $\theta^*\in\Theta_0$, $h(\xt{1}_i,\Theta_0)$ orthogonalizes $\Xt{2}_i$ against $\Xt{1}_i$.} This is the key to  Type I error control (see Theorem \ref{thm:t1e}). 
    \item[\emph{P2:}] \emph{When $\theta^*\not\in\Theta_0$, $h(\xt{1}_i,\Theta_0)$ fails to orthogonalize $\Xt{2}_i$ against $\Xt{1}_i$}.  The residual covariance between $\Xt{2}_i-h(\xt{1}_i,\Theta_0)$ and $\Xt{1}_i$ provides power to reject $H_0'(g)$ (see Section~\ref{subsec:power}).
\end{enumerate}
In practice, we never know whether $\theta^*\in\Theta_0$, and therefore must devise a single unified construction that guarantees Type I error control (i.e., \emph{P1}), whilst  simultaneously offering as much power as possible (i.e., \emph{P2}).

\begin{remark}
\label{rem:orth}
    At  first glance, it may be tempting to construct $h(\xt{1}_i,\Theta_0)$ by  regressing $\Xt{2}_i$ on $\Xt{1}_i$ using some flexible nonparametric regression method. Unfortunately, this will orthogonalize $\Xt{2}_i$ against $\Xt{1}_i$ under both the null $\theta^*\in\Theta_0$ \emph{and under the alternative} $\theta^*\not\in\Theta_0$, thereby eliminating all power to reject $H_0'(g)$. 
\end{remark}

Our choice of symmetric shift-family noise in Algorithm \ref{alg:split} creates structure in the conditional mean function $\E_{\eta_i(\theta^*)}[\Xt{2}_i|\Xt{1}_i=\xt{1}_i]$ (i.e., the first case of $h(\xt{1}_i,\Theta_0)$ in \eqref{eq:hcases}) that offers a path forward. 

\begin{proposition}
\label{prop:ratio}
Consider $\Xt{1}_i$ and $\Xt{2}_i$ constructed from $X_i$ and $W_i$ using Algorithm \ref{alg:split}; thus, $W_i\overset{iid}{\sim} R(0,\Sigma)$ for a symmetric shift-family $R$. Then, 
$$
\E_{\eta_i(\theta^*)}[\Xt{2}_i|\Xt{1}_i=\xt{1}_i] = 2\frac{N(\xt{1}_i)}{D(\xt{1}_i)}-\xt{1}_i,
$$
where $N(\xt{1}_i)$ and $D(\xt{1}_i)$ can be written as expectations taken with respect to a random variable $U_i\sim R(\xt{1}_i,\Sigma)$ that follows the distribution of the user-added noise $W_i$ shifted by $\xt{1}_i$:
\begin{equation}
N(\xt{1}_i) = \E_{U_i}[U_if_{X}(U_i;\eta_i(\theta^*))] \quad\text{and}\quad D(\xt{1}_i) = \E_{U_i}[f_{X}(U_i;\eta_i(\theta^*))], \label{eq:Uratio}
\end{equation}
where $f_{X}(\cdot;\eta_i(\theta^*))$ is the probability density/mass function of $X_i\sim F(\eta_i(\theta^*))$. 

Alternatively, they can be written as expectations taken with respect to $X_i$:
\begin{equation}
N(\xt{1}_i) = \E_{\eta_i(\theta^*)}[X_if_R(X_i;\xt{1}_i,\Sigma)] \quad\text{and}\quad D(\xt{1}_i) = \E_{\eta_i(\theta^*)}[f_R(X_i;\xt{1}_i,\Sigma)], \label{eq:Xratio}
\end{equation}
where $f_R(\cdot;\xt{1}_i,\Sigma)$ is the density/mass function of $U_i$.
\end{proposition}

Let $\hat N_n(\xt{1}_i)$ and $\hat D_n(\xt{1}_i)$ denote estimates of $N(\xt{1}_i)$ and $ D(\xt{1}_i)$ \emph{subject to the constraints imposed by $H_0: \theta^* \in \Theta_0$}; we discuss three possible options for constructing such estimates using \eqref{eq:Uratio} and \eqref{eq:Xratio} in Section \ref{subsec:estimation}. In light of Proposition \ref{prop:ratio}, we propose the following function $\hat h_n(\xt{1}_i,\Theta_0)$ as an estimate of $h(\xt{1}_i,\Theta_0)$: 
\begin{equation}
\label{eq:hstar}
\hat h_n(\xt{1}_i,\Theta_0) = 2\frac{\hat N_n(\xt{1}_i)}{\hat D_n(\xt{1}_i)}-\xt{1}_i.
\end{equation}

Recall that the function $h(\xt{1}_i,\Theta_0)$ must accomplish both \emph{P1} and \emph{P2} above: that is, it must orthogonalize $\Xt{2}$ against $\Xt{1}$ when $H_0$ holds, but not under the alternative. The function \eqref{eq:hstar} is attractive as it accomplishes both of these tasks \emph{without  knowledge of whether $H_0$ holds}.  The key is that $\hat N_n(\xt{1}_i)$ and $\hat D_n(\xt{1}_i)$ are estimated subject to the constraints imposed by $H_0$. Intuitively, if $H_0: \theta^* \in \Theta_0$ holds and if $\hat N_n(\xt{1}_i)$ and $\hat D_n(\xt{1}_i)$ are ``good enough" estimators of $N(\xt{1}_i)$ and $D(\xt{1}_i)$, respectively, then $\hat h_n(\xt{1}_i,\Theta_0)$ will successfully approximate $\E_{\eta_i(\theta^*)}[\Xt{2}_i|\Xt{1}_i=\xt{1}_i]$, and therefore will (approximately) orthogonalize $\Xt{2}_i$ against $\Xt{1}_i$, as required by \emph{P1}. If instead $\theta^*\not\in\Theta_0$, then since  $\hat N_n(\xt{1}_i)$ and $\hat D_n(\xt{1}_i)$ are estimated under $H_0: \theta^* \in \Theta_0$, they may be poor estimates of $N(\xt{1}_i)$ and $D(\xt{1}_i)$, respectively, in which case $\hat h_n(\xt{1}_i,\Theta_0)$ will fail to orthogonalize $\Xt{2}_i$ against $\Xt{1}_i$, as required by \emph{P2}.

\subsection{Testing $H_0'(g)$}
\label{subsec:test}

Let $M_1,\dots,M_K$ denote a random partition of the integers $\{1,\dots,n\}$ into $K$ approximately equally-sized subsets. Consider the following cross-fit sample average of outer products,
\begin{equation}
C_n(\xt{1},\xt{2};\Theta_0)=\frac{1}{n}\sum_{k=1}^K\sum_{i\in M_k}\left(\xt{1}_i+\xt{2}_i-2\frac{\hat N_n^{(-k)}(\xt{1}_i)}{\hat D_n^{(-k)}(\xt{1}_i)}\right)g(\xt{1}_i)^\top, 
\label{eq:crossfit}
\end{equation}
where $\hat N_n^{(-k)}(\xt{1}_i)$ and $\hat D_n^{(-k)}(\xt{1}_i)$ are estimates of $N(\xt{1}_i)$ and $D(\xt{1}_i)$, respectively, computed using only the observations $x_{i'}$ with $i'\not\in M_k$ \citep{chernozhukov2022locally}. 
$C_n(\xt{1},\xt{2};\Theta_0)$ can be understood as the sample analog of the expectation in \eqref{eq:H0prime}, with an out-of-sample variant of $\hat h_n(\xt{1}_i,\Theta_0)$ in \eqref{eq:hstar} used in place of $h(\xt{1}_i,\Theta_0)$. Taking inspiration from \cite{diciccio2017robust},  we will use the studentized squared $\ell_2$-norm of $C_n(\xt{1},\xt{2};\Theta_0)$ as a test statistic for  $H_0'(g)$. That is, we define 
\begin{equation}
\label{eq:Tn}
T_n(\xt{1},\xt{2};\Theta_0,\hat\Psi_n)=n\vect{C_n(\xt{1},\xt{2};\Theta_0)}^\top\hat\Psi_n^{-1}\vect{C_n(\xt{1},\xt{2};\Theta_0)},
\end{equation} 
where $\hat\Psi_n$ is an estimate of the asymptotic covariance of $\sqrt{n}\vect{C_n(\xt{1},\xt{2};\Theta_0)}$ that will be specified shortly. We will show that under mild conditions, $T_n(\xt{1},\xt{2};\Theta_0,\hat\Psi_n)$ asymptotically follows a $\chi_{pq}^2$ distribution, where $p$ is the dimension of $X$ and $q$ is the dimension of $g(\cdot)$. 

We begin our study of $T_n(\xt{1},\xt{2};\Theta_0,\hat\Psi_n)$ with Lemma \ref{lem:ASL}, which provides conditions under which $C_n(\Xt{1},\Xt{2};\Theta_0)$ is asymptotically linear. Our primary requirement is that $\hat N_n^{(-k)}(\xt{1}_i)$ and $\hat D_n^{(-k)}(\xt{1}_i)$ are asymptotically linear estimators of $N(\xt{1}_i)$ and $D(\xt{1}_i)$, respectively. 

\begin{lemma}
\label{lem:ASL}
Suppose that $H_0'(g)$ is true and that $\hat N_n^{(-k)}(s)$ and $\hat D_n^{(-k)}(s)$ are asymptotically linear estimators of $N(s)$ and $D(s)$ with influence functions $\varphi_N(\xt{1},\xt{2};s)$ and $\varphi_D(\xt{1},\xt{2};s)$, respectively. That is, for all $s\in\mathcal{X}^{(1)}$ and $k\in 1,\dots,K$, we can write:
\begin{align*}
\hat N_n^{(-k)}(s) - N(s) &= \frac{1}{n-|M_k|}\sum_{i\not\in M_k} \varphi_N(\xt{1}_i,\xt{2}_i;s) + o_P(n^{-1/2}), \\
\hat D_n^{(-k)}(s) - D(s) &= \frac{1}{n-|M_k|}\sum_{i\not\in M_k} \varphi_D(\xt{1}_i,\xt{2}_i;s) + o_P(n^{-1/2}).
\end{align*}

Then $C_n(\xt{1},\xt{2};\Theta_0)$ is also asymptotically linear with influence function
\small
\begin{equation} \label{eq:varphi-C}
\varphi_C(\xt{1}_i,\xt{2}_i)= \left(\xt{1}_i+\xt{2}_i-2\frac{N(\xt{1}_i)}{D(\xt{1}_i)}\right)g(\xt{1}_i)^\top -2\int\varphi_{ND}(\xt{1}_i,\xt{2}_i;\xt{1})g(\xt{1})^\top dP(\xt{1}),
\end{equation}
\normalsize
where $\varphi_{ND}(\xt{1}_i,\xt{2}_i;\cdot)=\left\{\varphi_N(\xt{1}_i,\xt{2}_i;\cdot)-\frac{N(\cdot)}{D(\cdot)}\varphi_D(\xt{1}_i,\xt{2}_i;\cdot)\right\}\frac{1}{D(\cdot)}$.
\end{lemma}

Theorem \ref{thm:Cn} then states that under standard conditions on $\varphi_C(\xt{1}_i,\xt{2}_i)$, $\vect{C_n(\Xt{1},\Xt{2};\Theta_0)}$ is asymptotically Gaussian. We assume that our data are generated as a triangular array, but we suppress indexing by $n$ for notational simplicity. 

\begin{theorem}[Asymptotic Gaussianity of $\sqrt{n}\vect{C_n(\Xt{1},\Xt{2};\Theta_0)}$]
\label{thm:Cn}
In the setting of Lemma \ref{lem:ASL}, suppose that for all $i$,
$\varphi_C(\Xt{1}_i,\Xt{2}_i)$ has finite variance, and define
$$
\Psi_n=\frac{1}{n}\sum_{i=1}^n \Var[\varphi_C(\Xt{1}_i,\Xt{2}_i)].
$$    
Suppose further that for every $\epsilon>0$, 
$$
\sum_{i=1}^n \E\left[\lVert \Psi_n^{-1/2}\varphi_C(\Xt{1}_i,\Xt{2}_i)\rVert^2I\left\{\lVert \Psi_n^{-1/2}\varphi_C(\Xt{1}_i,\Xt{2}_i)\rVert > \epsilon\sqrt{n}\right\}\right]\to 0.
$$ 
It follows that
$$
\sqrt{n}\Psi_n^{-1/2}\vect{C_n(\Xt{1},\Xt{2};\Theta_0)} \overset{D}\to N_{pq}(0,I_{pq}),
$$
where $p$ is the dimension of $X$ and $q$ is the dimension of $g(\cdot)$.
\end{theorem}

\begin{remark}
If $X_i$ are independent and identically distributed under $H_0$, then the classical central limit theorem applies and only the finite variance assumption is necessary \citep{van2000asymptotic}.
\end{remark}

\begin{corollary}
\label{cor:Tn}
In the setting of Theorem \ref{thm:Cn}, let $\hat\Psi_n$ denote a consistent estimator of $\Psi_n$. Then, $T_n(\Xt{1},\Xt{2};\Theta_0,\hat\Psi_n)\overset{D}{\to}\chi^2_{pq}$. 
\end{corollary}
\begin{remark}
Suppose that $X_i$ are independent and identically distributed under $H_0$. Then the sample variance of $\hat\varphi_C(\xt{1}_i,\xt{2}_i)$, 
the plug-in estimator of $\varphi_C(\xt{1}_i,\xt{2}_i)$ in \eqref{eq:varphi-C} using $\hat N_n^{(-k)}(\cdot)$ and $\hat D_n^{(-k)}(\cdot)$, is a consistent estimator of $\Psi_n$ for use in Corollary~\ref{cor:Tn}. 
\end{remark}

It follows from Corollary \ref{cor:Tn} that to test $H_0'(g)$, we can compare $T_n(\xt{1},\xt{2};\Theta_0,\hat\Psi_n)$ to the quantiles of the $\chi_{pq}^2$ distribution. Formally, we define 
\begin{equation}
\label{eq:test}
\vartheta_\alpha(\xt{1},\xt{2})=I\left(T_n(\xt{1},\xt{2};\Theta_0,\hat\Psi_n) \ge q_{pq,1-\alpha}\right)
\end{equation}
where $q_{pq,1-\alpha}$ is the $(1-\alpha)$-quantile of the $\chi_{pq}^2$ distribution. Proposition \ref{prop:asym} states that this choice of $\vartheta_\alpha(\Xt{1},\Xt{2})$ asymptotically controls the Type I error of $H_0$. It follows from the fact that $\vartheta_\alpha(\Xt{1},\Xt{2})$ is an asymptotically valid test of $H_0'(g)$ and that $H_0$ implies $H_0'(g)$.

\begin{proposition}
\label{prop:asym}
In the setting of Theorem \ref{thm:Cn}, $\vartheta_\alpha(\Xt{1},\Xt{2})$ as defined in \eqref{eq:test} is an asymptotically valid test of $H_0$ at significance level $\alpha$. That is, if $H_0$ is true, then for any $\alpha\in(0,1)$, 
$$
\limsup_{n\to\infty} P(\vartheta_\alpha(\Xt{1},\Xt{2}) = 1) \le \alpha.
$$
\end{proposition}

\subsection{Estimating $N(\xt{1}_i)$ and $D(\xt{1}_i)$ under $H_0$}
\label{subsec:estimation}

We now discuss how to obtain  asymptotically linear estimators of $N(\xt{1}_i)$ and $D(\xt{1}_i)$ under $H_0$, as required in Section~\ref{subsec:test}. We offer three strategies.  
The first two use the representation in \eqref{eq:Uratio}, and apply in the setting of a point null hypothesis (Section~\ref{subsubsec:point}) or in the setting where we have access to an asymptotically linear estimator  that  is consistent for  $\theta^*$  if  $H_0$ holds (Section~\ref{subsubsec:ale}). 
The third uses the representation in \eqref{eq:Xratio} and is applicable in cases where $H_0$ can be recast as testing whether $X_1,\ldots,X_n$ are independent and identically distributed (Section~\ref{subsubsec:iid}).

\subsubsection{Point null hypotheses}
\label{subsubsec:point}

In the setting of a point null hypothesis $H_0: \theta^*=\theta_0$, we can simply plug in $\theta_0$ in place of $\theta^*$ in \eqref{eq:Uratio} and compute $N(\xt{1}_i)$ and $D(\xt{1}_i)$ with a simple Monte Carlo simulation. Letting $u_1,\ldots,u_{B}$  denote independent realizations of $U_{i}\sim R(\xt{1}_{i},\Sigma)$, we have that as $B\to\infty$,
\begin{equation}
\label{eq:montecarlo}
\frac{1}{B}\sum_{b=1}^Bu_bf_{X}(u_b;\eta_{i}(\theta_0)) \overset{P}{\to} N(\xt{1}_{i}) \quad\text{and}\quad \frac{1}{B}\sum_{b=1}^{B}f_{X}(u_b;\eta_{i}(\theta_0)) \overset{P}{\to}  D(\xt{1}_{i}).
\end{equation}
Because the value of $B$ is in our control and the distribution of $U_{i}$ is known exactly --- i.e. we can generate as many realizations of $U_{i}$ as we want (subject to computational constraints) --- in effect, under $H_0$, $ N(\xt{1}_{i})$ and $D(\xt{1}_{i})$ are known exactly. Of course,  more efficient sampling schemes to approximate $N(\xt{1}_{i})$ and $D(\xt{1}_{i})$ are available.

\begin{remark}
If $R$ is a discrete distribution, then Monte Carlo estimation is not necessary. Rather, the expectations in \eqref{eq:Uratio} can be computed from first principles by summing over the support of $R(\cdot,\Sigma)$.
\end{remark}

\begin{remark}
In the case of a point null hypothesis, the testing procedure described in Section \ref{subsec:test} simplifies considerably because $N(\Xt{1}_i)$ and $D(\Xt{1}_i)$ are known exactly. It follows that the variance of $C_n(\Xt{1},\Xt{2};\Theta_0)$  in \eqref{eq:crossfit} takes a very simple form, namely
$$
\Var\left(C_n(\Xt{1},\Xt{2};\Theta_0)\right)=\frac{1}{n}\sum_{i=1}^n \Var\left[\left(\Xt{1}_i+\Xt{2}_i-2\frac{N(\Xt{1}_i)}{D(\Xt{1}_i)}\right)g(\Xt{1}_i)^\top\right].
$$
\end{remark}

The next two subsections consider the setting of a composite null hypothesis.

\subsubsection{Parametric and other settings admitting asymptotically linear estimators of $\theta^*$}
\label{subsubsec:ale}

Consider an asymptotically linear estimator, $\hat\theta_n$, computed from $x_1,\dots,x_n$, with influence function $\varphi_\theta(x)$, that is consistent for $\theta^*$ if $H_0$ holds. In parametric problems, under standard regularity conditions, the maximum likelihood estimator is one such example \citep{van2000asymptotic, bickel2001inference}. 

For each $k=1,\dots,K$, let $\hat\theta_n^{(-k)}$ denote the estimator $\hat\theta_n$ computed using all of the observations not in the $k$th fold. 
Equation \ref{eq:Uratio} in Proposition \ref{prop:ratio} suggests defining $\hat N_n^{(-k)}(\xt{1}_{i'})=\E_{U_{i'}}[U_{i'}f_{X}(U_{i'};\eta_{i'}(\hat\theta_n^{(-k)}))]$ and  $\hat D_n^{(-k)}(\xt{1}_{i'}) = \E_{U_{i'}}[f_{X}(U_{i'};\eta_{i'}(\hat\theta_n^{(-k)}))]$ for ${i'}\in M_k$. These quantities can be computed with a simple Monte Carlo simulation resembling \eqref{eq:montecarlo} with $\hat\theta_n^{(-k)}$ in place of $\theta_0$. 
Once again, since the value of $B$ can be arbitrarily large, 
given $\hat\theta_n^{(-k)}$, $\hat N_n^{(-k)}(\xt{1}_{i'})$ and $\hat D_n^{(-k)}(\xt{1}_{i'})$ are effectively known exactly.

It remains to show that $\hat N_n^{(-k)}(\xt{1}_{i'})$ and $\hat D_n^{(-k)}(\xt{1}_{i'})$ are asymptotically linear, so that they can be used in the context of Lemma \ref{lem:ASL}. This is established by the following result.
\begin{proposition}\label{prop:delta}
Assume that for $i'=1,\ldots,n$, $f_X(x;\eta_{i'}(\theta))$ is differentiable with respect to $\theta$ at $\theta^*$ and $\nabla_\theta f_X(x;\eta_{i'}(\theta)) \mid_{\theta=\theta^*} \ne 0$ 
for all $x\in\mathcal{X}$. Then for $i'\in M_k$, 
\begin{align*}
\hat N_n^{(-k)}(\xt{1}_{i'}) - N(\xt{1}_{i'}) &= \frac{1}{n-|M_k|}\sum_{i\not\in M_k}\E_{U_{i'}}[U_{i'}\nabla_\theta f_X(x;\eta_{i'}(\theta))^\top \mid_{\theta=\theta^*}] \varphi_\theta(x_i) + o_P(n^{-1/2}), \\
\hat D_n^{(-k)}(\xt{1}_{i'}) - D(\xt{1}_{i'}) &= \frac{1}{n-|M_k|}\sum_{i\not\in M_k}\E_{U_{i'}}[\nabla_\theta f_X(x;\eta_{i'}(\theta))^\top \mid_{\theta=\theta^*}] \varphi_\theta(x_i) + o_P(n^{-1/2}).\end{align*} 
\end{proposition}

In the special case where $g$ is chosen to be a scalar function, the following corollary provides a useful simplification of $\varphi_C(\cdot)$, defined in \eqref{eq:varphi-C}, in the setting of this subsection.

\begin{corollary}
\label{cor:delta}
In the setting of Proposition \ref{prop:delta}, suppose that $g$ is scalar-valued. Then,
$$
\varphi_C(\xt{1}_i,\xt{2}_i) = \left(\xt{1}_i+\xt{2}_i-2\frac{N(\xt{1}_i)}{D(\xt{1}_i)}\right)g(\xt{1}_i) -2A(g)\varphi_\theta(x_i)
$$
where $A(g)=\int\left\{\E_{U_s}[U_s\nabla_\theta f_X(x;\eta_{i'}(\theta))^\top \mid_{\theta=\theta^*}] 
-\frac{N(s)}{D(s)}\E_{U_{s}}[\nabla_\theta f_X(x;\eta_{i'}(\theta))^\top \mid_{\theta=\theta^*}] \right\}\frac{g(s)}{D(s)} dP_{\xt{1}}(s)$ and $U_s\sim R(s;\Sigma)$.
\end{corollary}

\subsubsection{Settings in which $X_i$ are independent and identically distributed under $H_0$}
\label{subsubsec:iid}

In the event that $X_1,\ldots,X_n$ are independent and identically distributed only under $H_0$ (e.g., under $H_0$ all observations are drawn from the same distribution, whereas under the alternative they are not), we can estimate $N(\xt{1}_i)$ and $D(\xt{1}_i)$ in \eqref{eq:Xratio} with sample means. Specifically, for $i'\in M_k$,  
\begin{equation}
\hat N_n^{(-k)}(\xt{1}_{i'}) = \frac{1}{n-|M_k|}\sum_{i\not\in M_k} x_if_R(x_i;\xt{1}_{i'},\Sigma) \quad\text{and}\quad \hat D_n^{(-k)}(\xt{1}_{i'}) = \frac{1}{n-|M_k|}\sum_{i\not\in M_k} f_R(x_i;\xt{1}_{i'},\Sigma). \label{eq:Xratio_est}
\end{equation}
These estimators require no distributional assumptions on $X_i$ and are linear by construction. Thus, Lemma~\ref{lem:ASL} applies directly. 

\subsection{Power of $\vartheta_\alpha(\xt{1},\xt{2})$}
\label{subsec:power}

We conclude this section by studying the power of $\vartheta_\alpha(\xt{1},\xt{2})$ in \eqref{eq:test}. To streamline our discussion, we consider testing $H_0:\theta^*=\theta_0$ against $H_1:\theta^*=\theta_1$. Because $H_0$ and $H_1$ are simple, their corresponding implied conditional means, $\E_{\eta_i(\theta_0)}[\Xt{2}_i|\Xt{1}_i=\xt{1}_i]$ and $\E_{\eta_i(\theta_1)}[\Xt{2}_i|\Xt{1}_i=\xt{1}_i]$, can be estimated to an arbitrary degree of precision by combining Proposition \ref{prop:ratio} with the strategy outlined in \eqref{eq:montecarlo}; for the remainder of this section we thus treat them as known.

Recall that testing $H_0$ via orthogonalization relies on the insight that $$
\sqrt{n}\Psi_n^{-1/2}\vect{C_n(\xt{1},\xt{2};\theta_0)}\overset{D}{\to}N_{pq}(0,I_{pq})
$$
under $H_0$, where $C_n(\xt{1},\xt{2};\theta_0)=\frac{1}{n}\sum_{i=1}^n\left(\xt{2}_i-\E_{\eta_i(\theta_0)}[\Xt{2}_i|\Xt{1}_i=\xt{1}_i]\right)g(\xt{1}_i)^\top$.

Suppose, however, that $H_1$ is instead true. In this case, observe that $C_n(\xt{1},\xt{2};\theta_0)$ can be decomposed into two terms:
\begin{align}
\begin{split}
\label{eq:Cn_alt}
C_n(\xt{1},\xt{2};\theta_0) &=\frac{1}{n}\sum_{i=1}^n\left(\xt{2}_i-\E_{\eta_i(\theta_1)}[\Xt{2}_i|\Xt{1}_i=\xt{1}_i]\right)g(\xt{1}_i)^\top \\
&\quad+\frac{1}{n}\sum_{i=1}^n\left(\E_{\eta_i(\theta_1)}[\Xt{2}_i|\Xt{1}_i=\xt{1}_i]-\E_{\eta_i(\theta_0)}[\Xt{2}_i|\Xt{1}_i=\xt{1}_i]\right)g(\xt{1}_i)^\top.
\end{split}
\end{align}

The first term is analogous to $C_n(\xt{1},\xt{2};\theta_0)$ under $H_0$, and converges to a mean-zero multivariate Gaussian (after appropriate rescaling). The second term is the source of power. It shifts the mean of $C_n(\xt{1},\xt{2};\theta_0)$ away from zero, which propagates forward into larger values of $T_n(\xt{1},\xt{2};\theta_0,\hat\Psi_n)$ in \eqref{eq:Tn} for an appropriate choice of $\hat\Psi_n$. Proposition \ref{prop:power} formalizes this intuition.

\begin{proposition}
\label{prop:power}
Consider  testing $H_0:\theta^*=\theta_0$ against $H_1:\theta^*=\theta_1$ and suppose that $H_1$ is true. Suppose that for all $i$,
$\left(\Xt{2}_i-\E_{\eta_i(\theta_0)}[\Xt{2}_i|\Xt{1}_i=\xt{1}_i]\right)g(\Xt{1}_i)^\top$ has finite variance and let
$$
\Psi_n=\frac{1}{n}\sum_{i=1}^n \Var\left[\left(\Xt{2}_i-\E_{\eta_i(\theta_0)}[\Xt{2}_i|\Xt{1}_i=\xt{1}_i]\right)g(\Xt{1}_i)^\top\right].
$$    
Suppose further that for every $\epsilon>0$,
\small{
$$
\frac{1}{n}\sum_{i=1}^n \E\left[\left\lVert \left(\Xt{2}_i-\E_{\eta_i(\theta_0)}[\Xt{2}_i|\Xt{1}_i=\xt{1}_i]\right)g(\Xt{1}_i)^\top\right\rVert^2I\left\{\left\lVert \left(\Xt{2}_i-\E_{\eta_i(\theta_0)}[\Xt{2}_i|\Xt{1}_i=\xt{1}_i]\right)g(\Xt{1}_i)^\top\right\rVert > \epsilon\sqrt{n}\right\}\right]\to 0.
$$}
It follows that
$$
\sqrt{n}\Psi_n^{-1/2}\left(\vect{C_n(\Xt{1},\Xt{2};\theta_0)}-\vect{\mu_n}\right) \overset{D}\to N_{pq}(0,I_{pq}),
$$
where $p$ is the dimension of $X$ and $q$ is the dimension of $g(\cdot)$, and where 
$$
\mu_n
=\frac{2}{n}\sum_{i=1}^n\E_{\eta_i(\theta_1)}\left[\left(\frac{\E_{U_i}[U_if_{X}(U_i;\eta_i(\theta_1))]}{\E_{U_i}[f_{X}(U_i;\eta_i(\theta_1))]}-\frac{\E_{U_i}[U_if_{X}(U_i;\eta_i(\theta_0))]}{\E_{U_i}[f_{X}(U_i;\eta_i(\theta_0))]}\right)g(\Xt{1}_i)^\top\right],
$$
where $U_i\sim R(\xt{1}_i,\Sigma)$ (see Proposition \ref{prop:ratio}).
\end{proposition} 

It follows from Proposition \ref{prop:power} that the power of $\vartheta_\alpha(\xt{1},\xt{2})$ is an increasing  function of $\lVert\Psi_n^{-1/2}\mu_n\rVert$. This quantity is in turn influenced by the choices of $g(\cdot)$ and $\Sigma$ in Algorithm \ref{alg:split}. The optimal choices of these hyperparameters will be context-specific. However, the following practical heuristics apply broadly:

\begin{itemize}
    \item To achieve power, the test functions $g(\cdot): \mathcal{X}^{(1)} \to \mathbb{R}^q$ should be collinear with the differences between $\E_{\eta_i(\theta_1)}[\Xt{2}_i|\Xt{1}_i=\xt{1}_i]$ and $\E_{\eta_i(\theta_0)}[\Xt{2}_i|\Xt{1}_i=\xt{1}_i]$. Richer test functions may therefore be preferable. However, care must be taken as increasing the dimension of $g(\cdot)$ also increases the difficulty of estimating $\Psi_n$. 
    \item If the noise variance is chosen to be too small, then  there will be very little power.
    To see this, note that in the extreme case where $\Sigma$ tends to zero, 
    $U_i \sim R(\xt{1}_i, \Sigma)$ is a point mass at $\xt{1}_i$. In this case, $\mu_n$ in Proposition~\ref{prop:power} equals zero, and thus there is no power to reject $H_0$. 
\end{itemize}

\section{Extension to post-selection inference}
\label{sec:sel}

\subsection{Background on post-selection  inference}
In recent years, the field of \emph{post-selection} or \emph{selective} inference has focused on hypothesis testing in the setting where the hypothesis is not fixed in advance, but rather, is itself a function of the data $X$. 
That is, we consider 
\begin{equation}
    \label{eq:h0_sel}
H_0(x): \theta^* \in \Theta_0(x).
\end{equation} 
In this setting, we would like to achieve \emph{selective} Type 1 error control, in the sense of \cite{fithian2014optimal}: that is, if $H_0(x)$ is true, then for any $\alpha\in(0,1)$, we want a test $\vartheta^{\Theta_0(x)}_\alpha(X)$ such that 
\begin{equation}\label{eq:sel-t1e}
P\left(\vartheta_\alpha^{\Theta_0(x)}(X) = 1 \mid \Theta_0(X)=\Theta_0(x)\right) \le \alpha.
\end{equation}
Here, the notation $\vartheta^{\Theta_0(x)}_\alpha(X)$ is intended to emphasize the fact that this is a test of the selected  hypothesis $H_0(x): \theta^* \in \Theta_0(x)$.

Several strategies towards achieving \eqref{eq:sel-t1e} have emerged in the literature.  One line of work, referred to as \emph{data carving} or \emph{conditional selective inference},  involves characterizing the distribution of a test statistic conditional on the event that this particular hypothesis was selected from the data \citep{fithian2014optimal, tian2018selective,lee2016exact}.  Despite its promise in a variety of contexts, including inference after variable selection  \citep{lee2016exact,panigrahi2024exact}, changepoint detection \citep{hyun2021post, jewell2022testing}, and clustering \citep{chen2022selective,yun2023selective,gao2020selective}, data carving approaches typically require the development of a new --- and typically quite complicated --- inferential procedure for every selection rule, and are largely restricted to multivariate Gaussian data. 

Another approach for inference on a selected hypothesis involves decomposing the data $X$ into two independent components $\Xt{1}$ and $\Xt{2}$, so that a hypothesis, $H_0(\xt{1}): \theta^* \in \Theta_0(\xt{1})$, can be selected using $\Xt{1}$ and then tested using $\Xt{2}$. For any $\alpha\in(0,1)$, when $H_0(\xt{1})$ holds, we then wish for a  guarantee along the lines of 
\begin{equation}\label{eq:sel-t1e-2}
P\left(\vartheta_\alpha^{\Theta_0 (\xt{1})} (\Xt{2}) = 1 \mid \Xt{1}=\xt{1} \right) \le \alpha,
\end{equation}
where the notation $\vartheta_\alpha^{\Theta_0 (\xt{1})} (\Xt{2})$ is intended to convey the fact that  $\Xt{2}$ is used to test a hypothesis that is a function of $\xt{1}$. 
Because $\Xt{1}$ and $\Xt{2}$ are independent, \emph{a test $\vartheta_\alpha^{\Theta_0 (\xt{1})} (\Xt{2})$ constructed as though the null hypothesis $H_0(\xt{1})$ were specified in advance will achieve \eqref{eq:sel-t1e-2}}. How can we decompose $X$ into independent components $\Xt{1}$ and $\Xt{2}$?
\emph{Sample splitting}   decomposes  $n$ independent and identically distributed random variables $X \sim Q^n$ into independent components $\Xt{1} \sim Q^{n_1}$ and $\Xt{2} \sim Q^{n_2}$ where $n_1+n_2=n$ \citep{cox1975note}. 
\emph{Data thinning} generalizes sample splitting  to settings where the latter cannot be applied, e.g. to settings where $n=1$, or where $n>1$ but the observations are not independent and identically distributed \citep{neufeld2023data, dharamshi2023generalized}. However, data thinning brings with it substantial distributional assumptions: i.e., only certain distributional families are amenable to thinning, and misspecification of the distributional family results in a loss of independence between $\Xt{1}$ and $\Xt{2}$.

\emph{Data fission} provides yet another approach: if $X$ is decomposed into \emph{dependent} components $\Xt{1}$ and $\Xt{2}$ in such a way that the conditional distribution of $\Xt{2} \mid \Xt{1}$ can be analytically characterized, then a guarantee along the lines of \eqref{eq:sel-t1e-2} can in some cases still be obtained \citep{leiner2022data}; however, due to dependence between $\Xt{1}$ and $\Xt{2}$, construction of a suitable test $\vartheta_\alpha^{\Theta_0 (\xt{1})} (\Xt{2})$ requires careful derivations involving  the conditional distribution of the test statistic given $\Xt{1}$. It is shown in \cite{dharamshi2024decomposing} that even in the very simple setting of a multivariate Gaussian distribution with unknown covariance, the required derivations are quite technical. Furthermore,  misspecification of the distributional family of $X$ will lead to a loss of downstream inferential guarantees. 

Thus, the aforementioned approaches offer practitioners a patchwork of solutions to the selective inference problem. It turns out, however, that the ideas developed in Sections~\ref{sec:intro}--\ref{sec:cmean} extend directly to the selective inference setting, thereby substantially expanding the set of selective inference problems for which solutions are available.  We provide details in the remainder of this section.

\subsection{Recasting a selected null hypothesis as a test of a moment condition}

To begin, we consider the following algorithm, which modifies Algorithm~\ref{alg:test} to allow for a test of a null hypothesis that is a function of $\xt{1}$. 

\begin{algorithm}[Post-selection inference via orthogonalization]
\label{alg:sel}
\textcolor{white}{.}

Input: Observed data $x_i$ drawn from $X_i\overset{ind}{\sim} F(\eta_i(\theta^*))$ for $i=1,\dots,n$;  a user-specified $p$-dimensional symmetric shift-family distribution $R(\phi,\Sigma)$ with mean $\phi$, user-specified variance $\Sigma$, and density $f_R(\cdot;\phi,\Sigma)$; and a test function $g:\mathcal{X}^{(1)}\rightarrow\mathbb{R}^q$.  
\begin{enumerate}
    \item Construct $\xt{1}_i$ and $\xt{2}_i$ using Algorithm \ref{alg:split} with inputs $x_1,\dots,x_n$ and $R(\phi,\Sigma)$. 
    \item Generate a hypothesis about $\theta^*$ using $\xt{1}$. Denote the selected null hypothesis as $H_0(\xt{1}):\theta^*\in\Theta_0(\xt{1})$ where $\Theta_0(\xt{1})\subset\Theta$.
    \item Test the hypothesis  
    \begin{equation}
\label{eq:H0prime2}
H_0'\left(\xt{1}, g\right): \E_{\eta_i(\theta^*)}\left[\left(\Xt{2}_i-h\left(\Xt{1}_i,\Theta_0(\xt{1}) \right)\right)g(\Xt{1}_i)^\top  \right]=0 \quad \forall i=1,\dots,n, 
\end{equation}
where $h\left(\Xt{1}_i,\Theta_0(\xt{1})\right)$ is defined as in \eqref{eq:hcases}. 
\end{enumerate}
\end{algorithm}

We now turn to the problem of testing $H_0'(\xt{1},g)$ in Step 3 of Algorithm~\ref{alg:sel}. To simplify the exposition, we will focus on cases in which under $H_0(\xt{1})$, $\theta^*$ admits a consistent and asymptotically linear estimator (i.e., the setting of Section \ref{subsubsec:ale}) and we will restrict our attention to scalar-valued $g$. The former condition is not restrictive, as established by \cite{tian2018selective}. We also note that Proposition~\ref{prop:ratio} holds in the selective setting: none of the statements in that proposition involve a null hypothesis, selective or otherwise.

Our approach is to construct a variant of $C_n(\xt{1},\xt{2};\Theta_0(\xt{1}))$ in \eqref{eq:crossfit} that is asymptotically linear and mean-zero conditional on $\Xt{1}$ when $H_0'(\xt{1},g)$ holds; an application of an appropriate conditional central limit theorem \citep{dedecker2002necessary, bulinski2017conditional, niu2024reconciling, zhao2025imputation} will then enable valid post-selection inference. 
In recent work, \cite{jin2024tailored} perform inference on asymptotically linear quantities conditional on components of the underlying data by debiasing the marginal influence function by its conditional mean. Their results suggest considering 
\begin{equation}
\label{eq:Cn'}
C_n'(\xt{1},\xt{2};\Theta_0(\xt{1})) = C_n(\xt{1},\xt{2};\Theta_0(\xt{1})) - \frac{1}{n}\sum_{i=1}^n \E_{\eta_i(\theta^*)}\left[\varphi_C(\Xt{1}_i,\Xt{2}_i)|\Xt{1}\right],
\end{equation}
where as in \eqref{eq:varphi-C}, $\varphi_C(\cdot)$ denotes the influence function of $C_n(\xt{1},\xt{2};\Theta_0(\xt{1}))$. While $C_n'(\xt{1},\xt{2};\Theta_0(\xt{1}))$ is conditionally mean-zero and asymptotically Gaussian, the  debiasing terms in \eqref{eq:Cn'} are not readily available as they are population quantities; they must instead be estimated. Proposition \ref{prop:cIF} provides a simplification of $\E_{\eta_i(\theta^*)}\left[\varphi_C(\Xt{1}_i,\Xt{2}_i)|\Xt{1}_i\right]$ starting from the representation in Corollary \ref{cor:delta}.

\begin{proposition}
\label{prop:cIF}
For any selected null hypothesis $H_0(\xt{1}):\theta^*\in\Theta_0(\xt{1})$, suppose that $H_0'(\xt{1},g)$ in \eqref{eq:H0prime2} is true, $\hat \theta_n^{(-k)}$ is an asymptotically linear estimator of $\theta^*$ with influence function $\varphi_\theta(x;\theta^*)$ so that Lemma \ref{lem:ASL} applies, and $g$ is a scalar-valued function. The notation $\varphi_\theta(x;\theta^*)$ is used to indicate that, in general, the influence function depends on the true value of the parameter $\theta^*$. Then,
$$
\E_{\eta_i(\theta^*)}\left[\varphi_C(\Xt{1}_i,\Xt{2}_i)|\Xt{1}\right] = -2A(g)\E_{\eta_i(\theta^*)}\left[\varphi_\theta(X_i;\theta^*)|\Xt{1}_i\right]=-2A(g)\frac{\E_{U_i}[\varphi_\theta(U_i;\theta^*)f_{X}(U_i;\eta_i(\theta^*))]}{\E_{U_i}[f_{X}(U_i;\eta_i(\theta^*))]},
$$
where $A(g)$ is defined in Corollary \ref{cor:delta}.
\end{proposition}

At a first glance, Proposition \ref{prop:cIF} seems to imply that one should estimate the debiasing terms in \eqref{eq:Cn'} with a plug-in estimator; that is, that we should consider
\begin{equation}
\label{eq:Cn1}
C_n^{(1)}(\Xt{1},\Xt{2};\Theta_0(\Xt{1})) = C_n(\Xt{1},\Xt{2};\Theta_0(\Xt{1})) + 2A(g) \frac{1}{n}\sum_{i=1}^n \frac{\E_{U_i}[\varphi_\theta(U_i;\hat\theta_n^{(-k)})f_{X}(U_i;\eta_i(\hat\theta_n^{(-k)}))]}{\E_{U_i}[f_{X}(U_i;\eta_i(\hat\theta_n^{(-k)}))]}
\end{equation}
where $\hat\theta_n^{(-k)}$ are cross-fit estimates of $\theta^*$. As before, we can compute all expectations taken with respect to $U_i$ using the Monte Carlo strategy outlined in Section \ref{subsubsec:ale}.

Perhaps surprisingly, $C_n^{(1)}(\Xt{1},\Xt{2};\Theta_0(\Xt{1}))$ in \eqref{eq:Cn1} is a poor choice. While debiasing using a sample average of plug-in influence function evaluations is a common tactic in semiparametric statistics \citep{kennedy2024semiparametric}, in the present setting, the plug-in bias correction is itself conditionally biased (i.e., the conditional mean of $C_n^{(1)}(\Xt{1},\Xt{2};\Theta_0(\Xt{1}))$ given $\Xt{1}$ is non-zero). 
Consider instead the following, which adjusts the plug-in debiasing term in \eqref{eq:Cn1} by a constant factor:
\begin{equation}
\label{eq:Dn}
D_n(\xt{1},\xt{2};\Theta_0(\xt{1})) = C_n(\Xt{1},\Xt{2};\Theta_0(\Xt{1})) + 2A(g)(I+B)^{-1}\frac{1}{n}\sum_{i=1}^n \frac{\E_{U_i}[\varphi_\theta(U_i;\hat\theta_n^{(-k)})f_{X}(U_i;\eta_i(\hat\theta_n^{(-k)}))]}{\E_{U_i}[f_{X}(U_i;\eta_i(\hat\theta_n^{(-k)}))]},
\end{equation}
where 
\begin{equation}
\label{eq:B}
B=\int\left(\nabla_\theta\frac{\E_{U_s}[\varphi_\theta(U_s;\theta)f_{X}(U_s;\eta_s(\theta))]}{\E_{U_s}[f_{X}(U_s;\eta_s(\theta))]}\Big{\vert}_{\theta=\theta^*}\right)^\top dP_{\xt{1}}(s)
\end{equation} and we require that $I+B$ is invertible. The next sequence of results confirms that $D_n(\xt{1},\xt{2};\Theta_0(\xt{1}))$ is conditionally mean-zero and asymptotically Gaussian. We start with Lemma \ref{lem:DnIF}, which provides a linearized representation of $D_n(\xt{1},\xt{2};\Theta_0(\xt{1}))$ in \eqref{eq:Dn}. As in previous sections, we assume our data are generated as a triangular array but suppress indexing by $n$ to simplify notation.

\begin{lemma}
\label{lem:DnIF}
In the setting of Proposition \ref{prop:cIF}, if $I+B$ is invertible, then $D_n(\xt{1},\xt{2};\Theta_0(\xt{1}))$ in \eqref{eq:Dn} can be written as
$$
D_n(\xt{1},\xt{2};\Theta_0(\xt{1}))= \frac{1}{n}\sum_{i=1}^n \varphi_D(\xt{1}_i,\xt{2}_i) + o_p(n^{-1/2})
$$
where
\small
\begin{align*}
\varphi_D(\Xt{1}_i,\Xt{2}_i)
&= \left(\Xt{1}_i+\Xt{2}_i-2\frac{N(\Xt{1}_i)}{D(\Xt{1}_i)}\right)g(\Xt{1}_i) -2A(g)(I+B)^{-1}\left(\varphi_\theta(X_i;\theta^*)-\E_{\eta_i(\theta^*)}\left[\varphi_\theta(X_i;\theta^*)|\Xt{1}_i\right]\right).
\end{align*}
\normalsize
\end{lemma}

\begin{remark}
\label{rem:geo}
The adjustment factor $(I+B)^{-1}$ in \eqref{eq:Dn} can be derived from a recursive sequence of debiasing steps that target $C_n'(\xt{1},\xt{2};\Theta_0(\xt{1}))$ in \eqref{eq:Cn'}. At each step, using $\hat\theta_n^{(-k)}$ in place of $\theta^*$ results in a term involving $\E_{\eta_i(\theta^*)}\left[\varphi_\theta(X_i;\theta^*)|\Xt{1}_i\right]$ (which is embedded in $\E_{\eta_i(\theta^*)}\left[\varphi_C(\Xt{1}_i,\Xt{2}_i)|\Xt{1}_i\right]$; see Proposition \ref{prop:cIF}). This, in turn, triggers a new bias term that must subsequently be removed. Remarkably, this sequence forms a geometric series with common ratio $-B$. Recall that such a series converges to $(I+B)^{-1}$ provided that $\max(|\lambda_{\max}(B)|, |\lambda_{\min}(B)|) \le 1 $. This is a stronger requirement than that required by Lemma \ref{lem:DnIF}, though it does hold in important special cases. We discuss the iterative debiasing perspective in detail in Supplement \ref{app:geometric}, then, in Supplement \ref{app:eigen}, we show that the corresponding eigenvalue condition holds when $\hat\theta_n^{(-k)}$ is an efficient estimator of $\theta^*$.
\end{remark}

The next theorem establishes that $D_n(\Xt{1},\Xt{2};\Theta_0(\Xt{1}))$ in \eqref{eq:Dn} is conditionally mean-zero and asymptotically Gaussian. The notion of convergence is complicated by conditioning. Here we define conditional convergence in distribution as convergence in probability of the conditional cumulative distribution function to the target cumulative distribution function; see \cite{niu2024reconciling} for a review of conditional asymptotics. 

\begin{theorem}[Conditional asymptotic Gaussianity of $\sqrt{n}\vect{D_n(\Xt{1},\Xt{2};\Theta_0(\Xt{1}))}$]
\label{thm:Dn}
Consider a sequence of selected hypotheses $H_{0,n}(\xt{1}):\theta^*\in\Theta_{0,n}(\xt{1})$ and their corresponding reformulated hypotheses $H_{0,n}'(\xt{1},g)$ (see \eqref{eq:H0prime2}). In the setting of Lemma \ref{lem:DnIF}, suppose that $g$ is a scalar-valued function chosen such that $\varphi_D(\cdot)$ is not identically zero. Define the filtration $\mathcal{F}_n=\sigma(\{\Xt{1}_i\}_{i=1}^n)$. Suppose that for all $i$,
$\Var\left[\varphi_D(\Xt{1}_i,\Xt{2}_i)|\mathcal{F}_n\right]$ is finite, and let
$$
\Omega_n=\frac{1}{n}\sum_{i=1}^n \Var\left[\varphi_D(\Xt{1}_i,\Xt{2}_i)|\mathcal{F}_n\right].
$$    
Suppose further that for every $\epsilon>0$, 
$$
\sum_{i=1}^n \E\left[\lVert \Omega_n^{-1/2} \varphi_{D}(\Xt{1}_i,\Xt{2}_i)\rVert^2I\left\{\lVert \Omega_n^{-1/2}\varphi_{D}(\Xt{1}_i,\Xt{2}_i)\rVert > \epsilon\right\}|\mathcal{F}_n\right]\to 0.
$$ 
It follows that for any $x\in\mathbb{R}^p$, under the sequence of null hypotheses given by $H_{0,n}'(\xt{1},g)$,
$$
P\left(\sqrt{n}\Omega_n^{-1/2}\vect{D_n(\Xt{1},\Xt{2};\Theta_0(\Xt{1}))} \le x \Big\vert\mathcal{F}_n \right)\overset{p}\to \Phi_p(x),
$$
where $\Phi_p$ is the cumulative distribution function of the $p$-dimensional standard multivariate Gaussian distribution and $p$ is the dimension of $X$. 
\end{theorem}

\begin{remark}
\label{rem:degenerate}
The restriction placed on $g$ in Theorem \ref{thm:Dn} guards against degeneracy in $\varphi_D$. From a practical perspective, most natural choices for $g$ do not lead to degeneracy in $\varphi_D$. The one exception is that in many settings, $g$ cannot be chosen to be constant in $\Xt{1}$ (though this choice tends to have lower power when it is available, and therefore is not advised regardless).
\end{remark}

\begin{remark}
Theorem \ref{thm:Dn} assumes knowledge of the constant $A(g)(I+B)^{-1}$ in the definition of $D_n(\xt{1},\xt{2};\Theta_0(\xt{1}))$ in \eqref{eq:Dn}, where $A(g)$ was defined in Corollary \ref{cor:delta} and $B$ was defined in \eqref{eq:B}. This is a function of $\theta^*$. In Supplement \ref{app:AB}, we show that using any consistent estimator of $A(g)(I+B)^{-1}$ has an asymptotically negligible impact on $D_n(\Xt{1},\Xt{2};\Theta_0(\Xt{1}))$. In practice, we estimate $A(g)(I+B)^{-1}$ with the cross-fit plug-in estimator with all integrals approximated with a Monte Carlo simulation.
\end{remark}

In light of Theorem \ref{thm:Dn}, we can perform post-selection inference with a $\chi^2$ test using  $D_n(\xt{1},\xt{2};\Theta_0(\xt{1}))$; this is analogous to the test for pre-specified hypotheses in Proposition \ref{prop:asym}. Consider 
\small
\begin{equation}
\label{eq:Tn2}
T_n(\Xt{1},\Xt{2};\Theta_0(\Xt{1}),\hat\Omega_n)=n\vect{D_n(\Xt{1},\Xt{2};\Theta_0(\Xt{1}))}^\top\hat\Omega_n^{-1}\vect{D_n(\Xt{1},\Xt{2};\Theta_0(\Xt{1}))},
\end{equation} 
\normalsize
where $\hat\Omega_n$ is a consistent estimator of the asymptotic covariance of $\sqrt{n}\vect{D_n(\Xt{1},\Xt{2};\Theta_0(\Xt{1}))}$. We define the test statistic
\begin{equation}
\label{eq:test2}
\vartheta_\alpha^ {\Theta_0(\xt{1})}\left(\Xt{1},\Xt{2}\right)=I\left(T_n\left(\Xt{1},\Xt{2};\Theta_0(\xt{1}),\hat\Omega_n\right) \ge q_{p,1-\alpha}\right),
\end{equation}
where $q_{p,1-\alpha}$ is the $(1-\alpha)$-quantile of the $\chi_{p}^2$ distribution. The notation $\vartheta_\alpha^ {\Theta_0(\xt{1})}\left(\Xt{1},\Xt{2}\right)$ is intended to convey the fact that the test involves both $\Xt{1}$ and $\Xt{2}$ (in contrast, for instance, to the test in \eqref{eq:sel-t1e-2}).

The next two corollaries to Theorem \ref{thm:Dn} describe the properties of $T_n(\Xt{1},\Xt{2};\Theta_0(\Xt{1}),\hat\Omega_n)$ and $\vartheta_\alpha^ {\Theta_0(\xt{1})}\left(\Xt{1},\Xt{2}\right)$. They follow from applications of the continuous mapping theorem and conditional Slutsky's theorem \citep{niu2024reconciling}.

\begin{corollary}
\label{cor:T2}
In the setting of Theorem \ref{thm:Dn}, let $\hat\Omega_n$ denote a consistent estimator of $\Omega_n$. Then,
for any $x\in\mathbb{R}^+$, under the sequence of null hypotheses given by $H_{0,n}'(\xt{1},g)$, it holds that
$$
P\left(T_n(\Xt{1},\Xt{2};\Theta_0(\Xt{1}),\hat\Omega_n)\le x\Big\vert\mathcal{F}_n \right)\overset{p}\to F_{p}(x)
$$
where $F_p$ is the cumulative distribution function of the $\chi^2_{p}$ distribution.
\end{corollary}

\begin{corollary}
\label{cor:selt1e}
In the setting of Theorem \ref{thm:Dn}, let $\hat\Omega_n$ denote a consistent estimator of $\Omega_n$. Then, 
for any $\alpha\in(0,1)$, under the sequence of null hypotheses given by $H_{0,n}'(\xt{1},g)$, it holds that 
$$
P\left(\vartheta_\alpha^ {\Theta_0(\xt{1})}\left(\Xt{1},\Xt{2}\right)=1\Big\vert\mathcal{F}_n \right)\overset{p}\to \alpha.
$$
\end{corollary}

We conclude with Theorem \ref{thm:sel}, which confirms that asymptotic Type I error control of $H_0'(\xt{1},g)$ in Corollary \ref{cor:selt1e} propagates back to $H_0(\xt{1})$. It follows immediately from the fact that $H_0(\xt{1})$ implies $H_0'(\xt{1},g)$.

\begin{theorem}[Validity of Algorithm \ref{alg:sel}]
\label{thm:sel}
Suppose that we use Algorithm \ref{alg:sel} to generate the hypothesis $H_0(\xt{1})$ and that $\vartheta_\alpha^ {\Theta_0(\xt{1})}\left(\Xt{1},\Xt{2}\right)$ is an asymptotically valid test for $H_0'(\xt{1},g)$ in the sense of Corollary \ref{cor:selt1e}.
Then, $\vartheta_\alpha^{\Theta_0(\xt{1})}\left(\Xt{1},\Xt{2}\right)$ is also an asymptotically valid test of $H_0(\xt{1})$.
\end{theorem}

\begin{remark}
The guarantee in Corollary \ref{cor:selt1e} (and Theorem \ref{thm:sel}) is slightly weaker than the guarantees outlined at the beginning of this section (see \eqref{eq:sel-t1e-2}), in the sense that our test offers convergence in probability to the nominal Type I error rate, rather than almost sure convergence.
\end{remark}

\section{Simulation studies}
\label{sec:cases}

In this section, we illustrate our proposal with two simulation studies. In Section \ref{subsec:twosample} we conduct a nonparametric test for a difference between two samples via orthogonalization; this is an example of a pre-specified hypothesis. In Section \ref{subsec:clustering} we conduct inference after clustering; this is an example of a post-selection inference scenario for which existing methods require strong distributional assumptions. Supplement \ref{app:cp} contains a third case study in which we conduct inference after changepoint detection via orthogonalization; a second post-selection inference example.

\subsection{Nonparametric two-sample testing}
\label{subsec:twosample}

Suppose that we observe two independent samples, $Y_i\overset{\text{iid}}{\sim}P$ for $i=1,\dots,n$ and $Z_j\overset{\text{iid}}{\sim}Q$ for $j=1,\dots,n$, where both $P$ and $Q$ are distributions defined on $\mathcal{X}$. Our goal is to test the pre-specified hypothesis $H_0:P=Q$. In this subsection we show how Algorithm \ref{alg:test} can be used to test $H_0$, \emph{ without imposing assumptions on $P$ and $Q$.}

Our approach follows from the observation that the combined dataset $X=(Y_1,\dots,Y_n,Z_1,\dots,Z_n)$ is independent and identically distribution if and only if $H_0$ holds. We can therefore test $H_0$ by testing whether $X$ forms an independent and identically distributed sample using the nonparametric strategy outlined in Section \ref{subsubsec:iid}.

We first consider a ``null" setting in which $P=Q=N_5(0_5,I_5)$. For $n\in\{250,500,1000,2500\}$, we draw 1,000 replicates of $Y$ and $Z$, then decompose $x$ into $\xt{1}$ and $\xt{2}$ using Algorithm \ref{alg:split} with $W_i\overset{\text{iid}}{\sim}N_5(0,cI_5)$ for $c\in\{2,5,10\}$. We then test $H_0$ by testing the reformulated $H_0'(g)$ in \eqref{eq:H0prime} with three choices of $g$: (i) $g(\xt{1}_i)=\xt{1}_{i1}$ (i.e., the first coordinate of $\xt{1}_i$), (ii) $g(\xt{1}_i)=\xt{1}_i$, and (iii) $g(\xt{1}_i)=\lVert \xt{1}_i \rVert_2$. Figure \ref{fig:ts-t1e} displays the p-values for this simulation study. Each panel corresponds to a value of $n$ and displays the corresponding empirical quantiles of the p-values against the quantiles of the uniform distribution. For $g(\xt{1}_i)=\xt{1}_{i1}$ and $g(\xt{1}_i)=\lVert \xt{1}_i \rVert_2$, the Type I error rate is controlled in all settings. For $g(\xt{1}_i)=\xt{1}_i$, larger sample sizes are required as the resulting $C_n(\Xt{1},\Xt{2};\Theta_0)$ has $5\times 5=25$ dimensions.

\begin{figure}[h]
\centering
\includegraphics[width=0.9\linewidth]{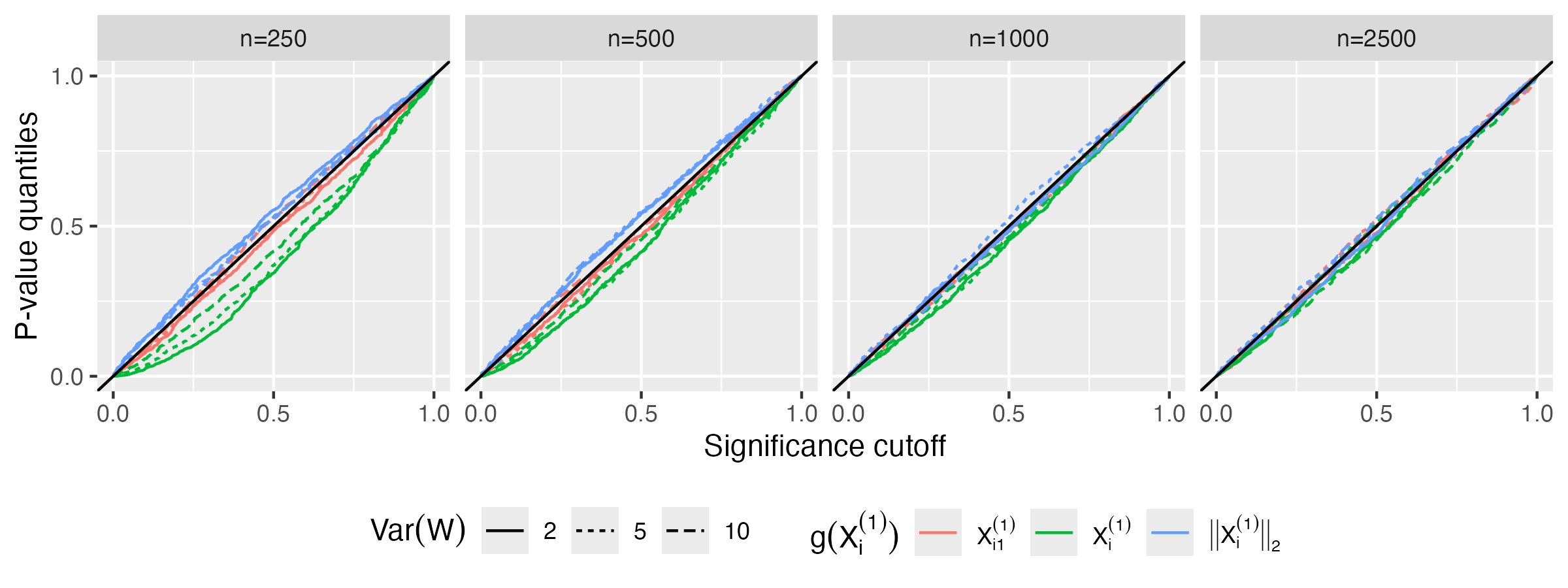}
\caption{Type I error results for the ``null" setting of the simulation described in Section \ref{subsec:twosample}. Each panel displays a QQ-plot of the empirical quantiles of the observed p-values against the quantiles of a $\text{Uniform}(0,1)$ distribution. For $g(\xt{1}_i)=\xt{1}_{i1}$ and $g(\xt{1}_i)=\lVert \xt{1}_i \rVert_2$, the Type I error rate is controlled for all settings of $n$ (indicated by panel) and values of $c$ (indicated by line type). For $g(\xt{1}_i)=\xt{1}_i$, the Type I error is controlled in larger sample sizes.} 
\label{fig:ts-t1e}
\end{figure}

Next, to assess power, we consider three ``alternative" settings in which $P=N_5(0_5,I_5)$ and $Q$ is set to: (a) $N_5(\delta\cdot 1_5,I_5)$ for $\delta\in\{1,2,3,4\}$; (b) $N_5(0_5,(1+\delta)\cdot I_5)$ for $\delta\in\{1,2,3,4\}$; or (c) $N_5(0_5,\delta 1_51_5^\top + (1-\delta)I_5)$ for $\delta\in\{0.2,0.4,0.6,0.8\}$. We refer to these alternatives as the ``mean", ``variance", and ``covariance" alternatives, respectively. We then test $H_0$ with the same choices of $n$, $c$, and $g$ used in the ``null" setting. Figure \ref{fig:ts-pow} displays the power of our approach as a function of $\delta$. Each panel corresponds to a setting of $n$ and choice of alternative. In the ``mean" alternative, power increases as a function of $n$ and $\delta$ and decreases with $c$; power is highest for the most expressive test function $g(\xt{1}_i)=\xt{1}_i$, and lowest for $g(\xt{1}_i)=\lVert \xt{1}_i \rVert_2$. For the ``variance" and ``covariance" alternatives, when $g(\xt{1}_i)=\xt{1}_{i1}$ or $g(\xt{1}_i)=\xt{1}_i$, power again increases as a function of $n$ and $\delta$ and decreases with $c$. By contrast, the coarser test function $g(\xt{1}_i)=\lVert \xt{1}_i \rVert_2$ does not have power against these alternatives, underscoring the importance of selecting an appropriate test function.

\begin{figure}[h]
\centering
\includegraphics[width=0.9\linewidth]{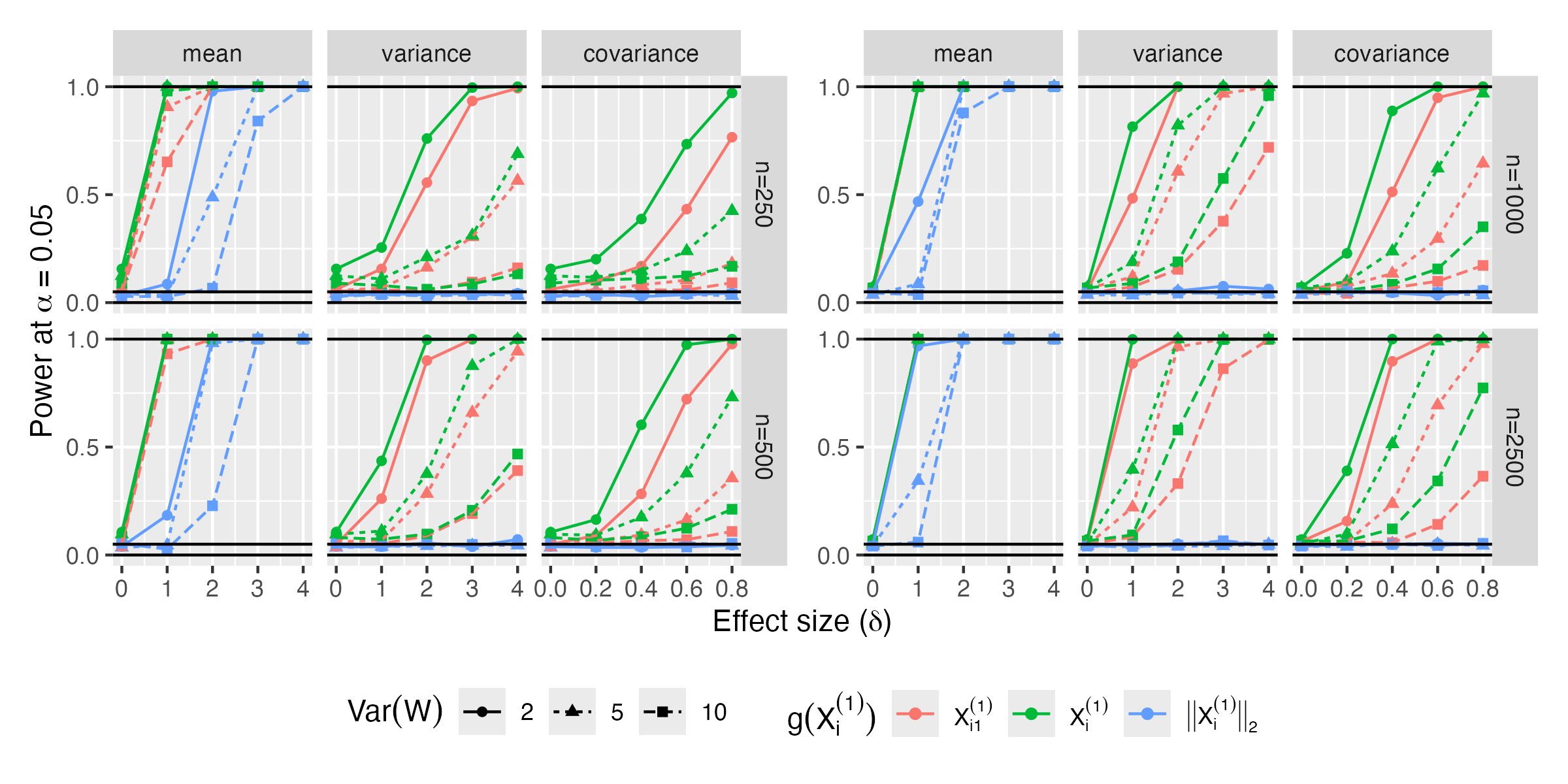}
\caption{Power results for the ``alternative" setting of the simulation described in Section \ref{subsec:twosample}. Each panel corresponds to a setting of $n$ and choice of alternative, and displays power curves as a function of the effect size $\delta$ for a choice of $g$ (indicated by colour) and value of $c$ (indicated by line type). Power grows with $n$ and $\delta$ and decreases with $c$; selecting $g(\xt{1}_i)=\lVert \xt{1}_i \rVert_2$ does not provide power against the ``variance" and ``covariance" alternatives.} 
\label{fig:ts-pow}
\end{figure}

\subsection{Inference after clustering}
\label{subsec:clustering}

As pointed out by \cite{gao2020selective}, testing for a difference between estimated clusters requires care: since the cluster assignments are a function of the data, classical tests are not valid; furthermore, sample splitting does not provide a valid solution. 
Here we consider the problem of conducting inference after clustering with zero-inflated Poisson data, often used to model single-cell RNA sequencing data (see Section \ref{sec:data}). This setting lies outside of the scope of existing approaches. Fortunately, Algorithm \ref{alg:sel}  provides a path forward.

More specifically, for $p\in\{2,10,50\}$ and $n\in\{100p,250p,500p\}$, we consider data generated as 
\begin{equation}
\label{eq:cluster}
X_{ij} \overset{\text{ind}}\sim \text{ZIP}(\lambda_{ij},\pi_{ij}),
\end{equation}
for $i=1,\dots,n$, where $\lambda_{ij}$ and $\pi_{ij}$ are the rate parameter and zero-inflation probability, respectively, for the $i$th observation in the $j$th coordinate.

We first consider a ``null" setting in which for all $i=1,\dots,n$, $\lambda_{ij}=3+j\% 3$ and $\pi_{ij}=0.2+0.05*(j \% 3)$, where the notation $a\% b$ indicates the remainder when $a$ is divided by $b$. As all observations share the same parameters, $X_1,\dots,X_n$ form an independent and identically distributed sample; in other words, there are no true subgroups to be discovered, and the null hypothesis holds for any set of estimated clusters. 

For each combination of $n$ and $p$, we draw $500$ replicates of $X$, and then apply a variant of Algorithm \ref{alg:sel} specialized to the present context. Briefly, for each replicate, we first apply Algorithm \ref{alg:split} with $W_{ij}\overset{\text{iid}}{\sim}\text{DiscreteUniform}(-3,3)$ to decompose $x$ into $\xt{1}$ and $\xt{2}$, then apply $k$-means clustering with $k=3$ to $\xt{1}$ to identify three clusters. Our goal is to test whether the data in the two largest estimated clusters are drawn from the same distribution; that is, to test the selected hypothesis $H_0(\xt{1}):(\lambda_{ij},\pi_{ij})=(\lambda_{i'j},\pi_{i'j}),\quad\forall i,i'\in \widehat{\mathcal{C}}_1\cup\widehat{\mathcal{C}}_2,j=1,\dots,p$, where $\widehat{\mathcal{C}}_1$ and $\widehat{\mathcal{C}}_2$ index the observations in the two largest estimated clusters. We test $H_0(\xt{1})$ by testing the reformulated $H_0'(\xt{1},g)$ with three choices of $g$: (i) $g(\xt{1}_i)=I(i \in \widehat{\mathcal{C}}_1)$, (ii) $g(\xt{1}_i)=\lVert \xt{1}_i \rVert_2I(i\in\widehat{\mathcal{C}}_1\cup\widehat{\mathcal{C}}_2)$, and (iii) $g(\xt{1}_i)=\lVert \xt{1}_i \rVert_\infty I(i\in\widehat{\mathcal{C}}_1\cup\widehat{\mathcal{C}}_2)$ (the indicators remove the dependence of $D_n(\Xt{1},\Xt{2};\Theta_0(\Xt{1}))$ on parameters corresponding to observations assigned to the third cluster). We do not consider the test function $g(\xt{1}_i)=I(i\in\widehat{\mathcal{C}}_1\cup\widehat{\mathcal{C}}_2)$ as it is analogous to a constant function, leading to degeneracy in $\varphi_D$; see Remark \ref{rem:degenerate}. Algorithm \ref{alg:cluster} in Supplement \ref{app:cluster} provides further details on this procedure.

Figure \ref{fig:clust-t1e} displays the p-values for this experiment. Each column corresponds to a value of $p$ and each row corresponds to a value of $n$. Each panel displays the corresponding empirical quantiles of the p-values against the quantiles of the uniform distribution. For all values of $n$, all values of $p$, and all choices of $g$, our procedure controls the Type I error rate. When $p=50$, our p-values are mildly conservative when $g$ is chosen as a norm, though this fades as the sample size increases.

\begin{figure}[h]
\centering
\begin{subfigure}{0.5\textwidth}
    \centering
    \includegraphics[width=0.9\linewidth]{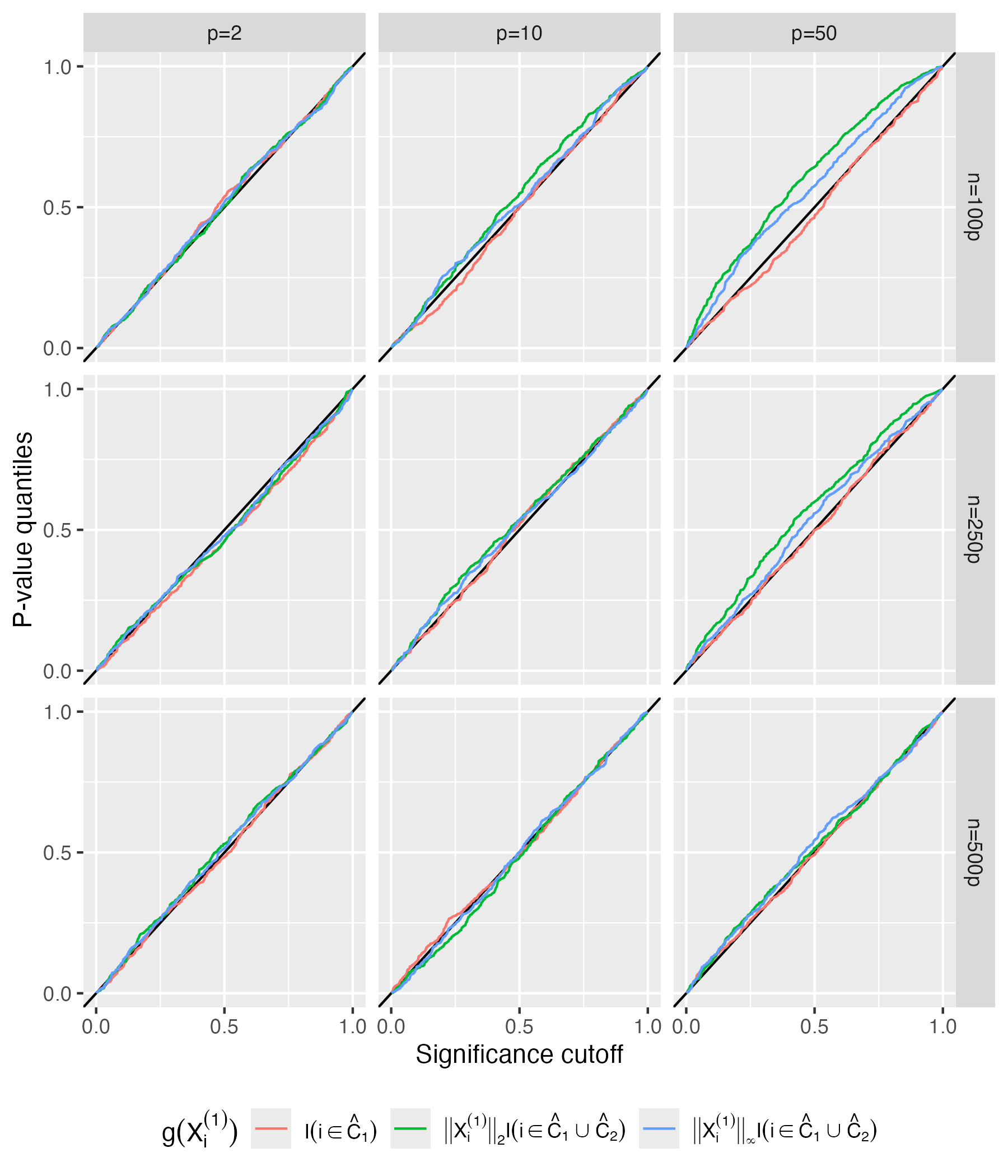}
    \caption{Type 1 error quantile-quantile plot \label{fig:clust-t1e}}
\end{subfigure}%
~
\begin{subfigure}{0.5\textwidth}
    \centering
    \includegraphics[width=0.9\linewidth]{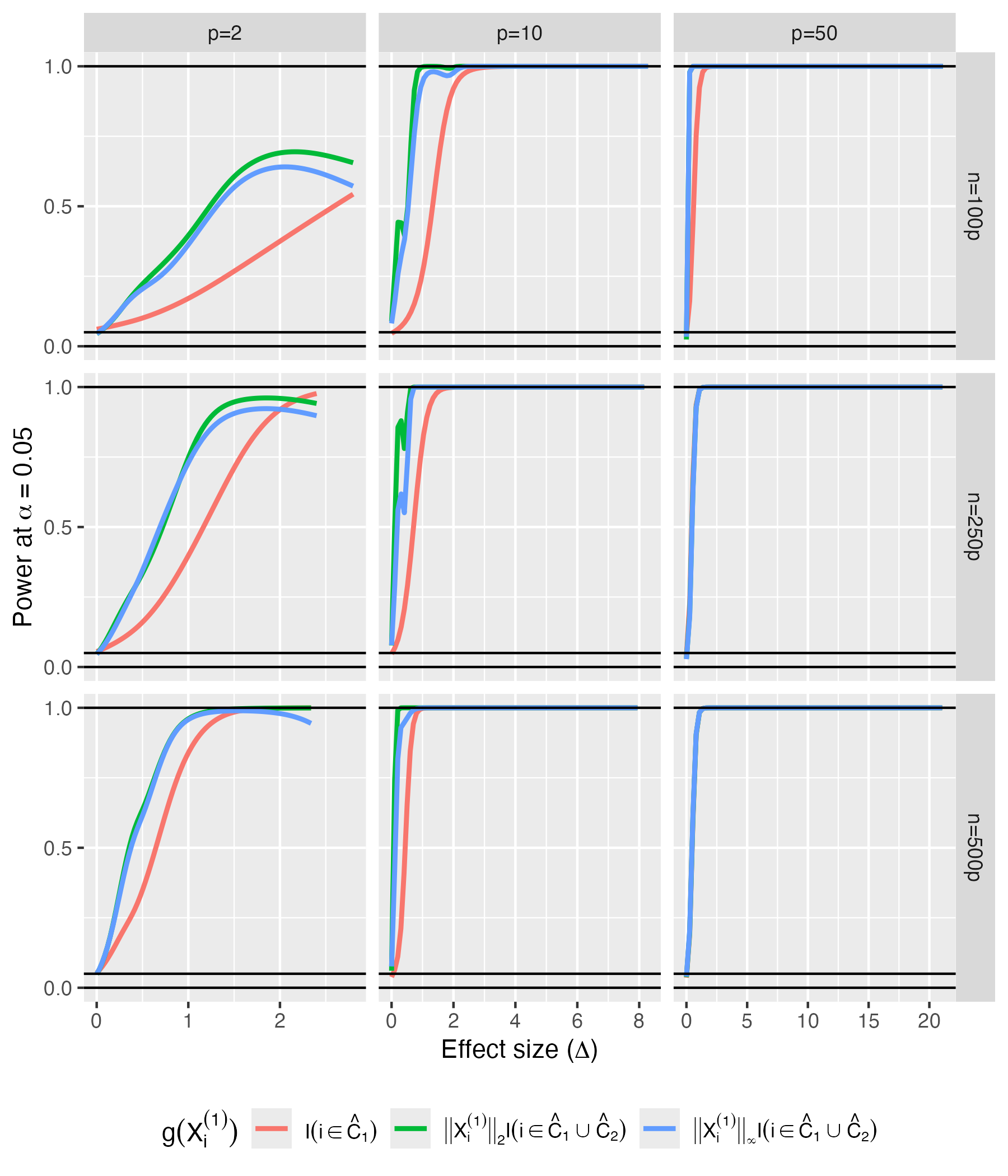}
    \caption{Smoothed power plot \label{fig:clust-pow}}
\end{subfigure}%
\caption{(a) Type I error results for the ``null" setting of the simulation described in Section \ref{subsec:clustering}. Each panel displays a QQ-plot of the empirical quantiles of the observed p-values against the quantiles of a $\text{Uniform}(0,1)$ distribution. The Type I error rate is controlled on average, across realizations of $\Xt{1}$, for all settings of $p$ (indicated by column), settings of $n$ (indicated by row), and choices of test function $g$ (indicated by colour). 
(b) Smoothed power results for the ``alternative" setting of the simulation described in Section \ref{subsec:clustering}. Each panel corresponds to a setting of $p$ (indicated by column) and $n$ (indicated by row), and displays three power curves as a function of the effect size $\Delta$ (defined in Section \ref{subsec:clustering}), one for each choice of $g$ (indicated by colour). In all cases, power increases as a function of $n$ and $\Delta$; power is also greater for  $g(\xt{1}_i)=\lVert \xt{1}_i \rVert_2I(i\in\widehat{\mathcal{C}}_1\cup\widehat{\mathcal{C}}_2)$ and $g(\xt{1}_i)=\lVert \xt{1}_i \rVert_\infty I(i\in\widehat{\mathcal{C}}_1\cup\widehat{\mathcal{C}}_2)$ as compared to $g(\xt{1}_i)=I(i \in \widehat{\mathcal{C}}_1)$.}
\end{figure}

The QQ-plots show that the Type I error for a test of $H_0(\xt{1})$ is controlled on average, across realizations of $\Xt{1}$. To verify Type I error control conditional on $\Xt{1}=\xt{1}$, as suggested by Corollary \ref{cor:selt1e} and Theorem \ref{thm:sel}, for $p=2$ we conduct a related experiment in which we first generate $500$ replicates of $\Xt{1}$, then for each replicate draw $500$ replicates of $\Xt{2}$ from the conditional distribution $\Xt{2}|\Xt{1}$. We then apply Steps 2--5 of Algorithm \ref{alg:cluster} to each of the $250,000$ pairs; Figure \ref{fig:ct1e-clust} displays the Type I error rate stratified by $\Xt{1}$ (i.e., the Type I error rate conditional on $\Xt{1}=\xt{1}$). The figure confirms that as the sample size grows, the conditional Type I error increasingly concentrates on the nominal level.

\begin{figure}[h]
\centering
\includegraphics[width=.8\linewidth]{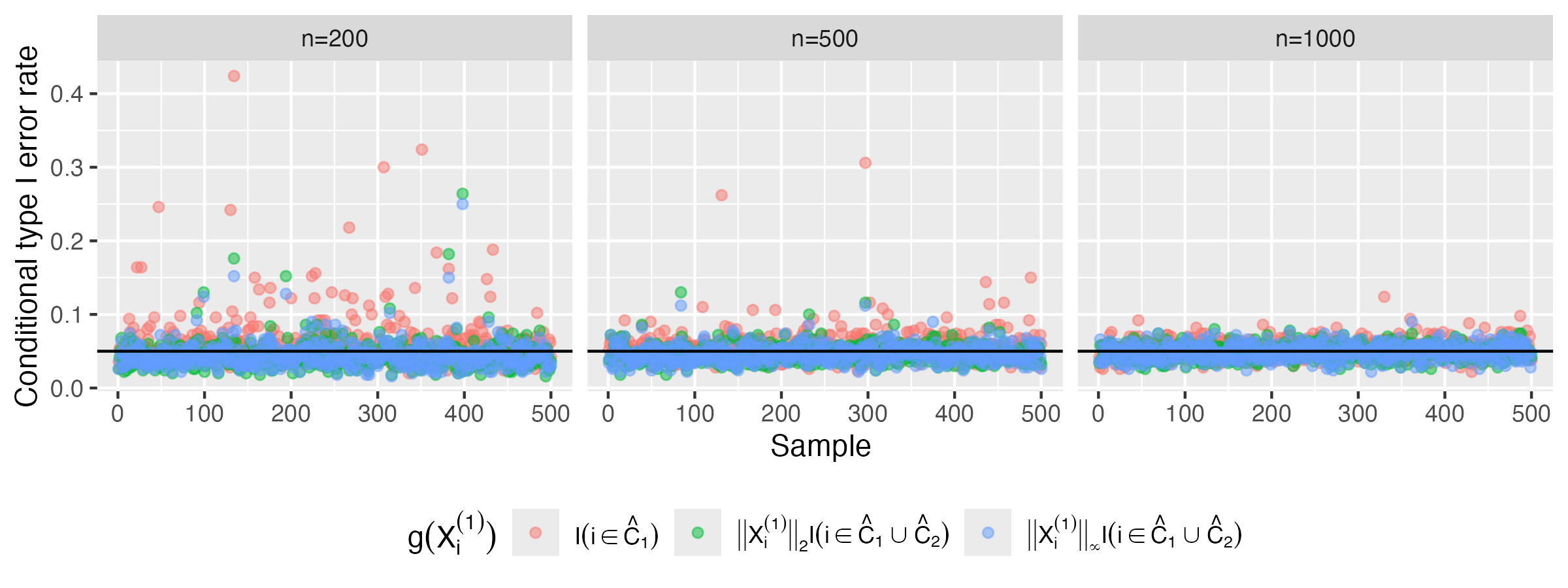}
\caption{Conditional Type I error results for the simulation described in Section \ref{subsec:clustering}. Each panel corresponds to a setting of $n$ and displays the empirical conditional Type I error rate for 500 draws of $\Xt{1}$. For all choices of $g$ (indicated by colour), the conditional Type I error rate increasingly concentrates on the nominal level $\alpha$ (indicated by the black line) as $n$ grows.}
\label{fig:ct1e-clust}
\end{figure}

To examine the power of our approach, we consider an ``alternative" setting in which for all  $i=1,\dots,n$, $\lambda_{ij}=3+d\cdot I(i \le n/2)+j\% 3$ and $\pi_{ij}=0.2+0.05*(j \% 3)$ and $d\in\{1,2,3,4\}$; i.e., there are two true subgroups of equal sizes with different rate parameters. For each combination of $n$, $p$, and $d$, we draw 500 replicates of $X$, then apply Algorithm \ref{alg:cluster}. Figure \ref{fig:clust-pow} displays the power of our approach as a function of the effect size, defined as $\Delta=\left\|\frac{1}{|\widehat{\mathcal{C}}_1|}\sum_{i\in\widehat{\mathcal{C}}_1}\E[X_i]-\frac{1}{|\widehat{\mathcal{C}}_2|}\sum_{i\in\widehat{\mathcal{C}}_2}\E[X_i]\right\|_2$ \citep{gao2020selective}. Each column corresponds to a value of $p$ and each row corresponds to a value of $n$. Each panel displays one (smoothed) power curve for each choice of $g$ (indicated by colour). The curves are constructed by fitting a regression spline using the \texttt{gam} function implemented in the \texttt{mgcv} R package \citep{mgcv}. As expected, power increases as a function of $n$ and $\Delta$. Interestingly, choosing $g$ to be either the $\ell_2$- or $\ell_\infty$-norm has higher power than the indicator for $\widehat{\mathcal{C}}_1$ membership; this is consistent with our intuition that more expressive choices of $g$ are preferable.
\section{Application to single-cell RNA-sequencing data}
\label{sec:data}

Single-cell RNA-sequencing (scRNA-seq) quantifies the gene expression profiles of individual cells, offering biologists fine-scale insight into cell characteristics and development. A common task in the analysis of scRNA-seq data is to cluster cells based on their gene expression profiles to identify candidate cell subtypes, and then to quantify uncertainty associated with those estimated clusters \citep{lahnemann2020eleven}; as discussed in Section \ref{subsec:clustering}, testing whether estimated subtypes are truly distinct requires accounting for the fact that they were estimated from the data.

The data from an scRNA-seq experiment presents as a count-valued matrix in which each row corresponds to a cell and each column corresponds to a gene. Due to limitations in the sensitivity of sequencing technology, scRNA-seq data often contain a higher than expected proportion of zeroes \citep{hicks2018missing}; this motivates the application of zero-inflated models \citep{jiang2022statistics, nguyen2023structure}. Here we illustrate how our proposal can be used to test for a difference in candidate cell subtypes identified from scRNA-seq data modelled with a zero-inflated Poisson distribution. 

In particular, we revisit the peripheral blood mononuclear cell (PBMC) data prepared by \cite{duo2020systematic}, which is a subset of the 68,000 PBMCs sequenced by \cite{zheng2017massively}. It contains gene expression profiles for 3,994 cells; each cell is furthermore annotated using a different technology as either a B-cell, CD14 monocyte, naive cytotoxic T-cell, or regulatory T-cell. Cell subtypes are present in roughly equal proportions. Treating these annotations as the truth, we assess the ability of our method to detect latent cell subtypes, then subsequently conduct inference.

We begin by subsetting the data to the 50 genes with the highest variance; let $X$ denote the resulting $3,994\times 50$ matrix. We assume the rows, $X_i$, follow the zero-inflated Poisson model in \eqref{eq:cluster}, then perform the following:
\begin{enumerate}
    \item Decompose $x$ into $\xt{1}$ and $\xt{2}$ using Algorithm \ref{alg:split} with $W_{ij}\overset{\text{iid}}{\sim}\text{DiscreteUniform}(-5,5)$; this choice of noise has approximately half the variance of the least variable gene in $X$.
    \item Apply $k$-means clustering to $\xt{1}_i/1_{50}^\top \xt{1}_i$ with $k=4$. Let $\hat{\mathcal{C}}_1,\dots,\hat{\mathcal{C}}_4$ index the cells in the four estimated clusters, respectively.
    \item For $l,l'$ satisfying $1\le l<l'\le 4$, test $H_0(\xt{1}):(\lambda_{ij},\pi_{ij})=(\lambda_{i'j},\pi_{i'j}),\quad\forall i,i'\in \hat{\mathcal{C}}_l\cup\hat{\mathcal{C}}_{l'},j=1,\dots,p$ by testing the resulting $H_0'(\xt{1},g)$ with $g(\xt{1}_i)=I(i\in\hat{\mathcal{C}_l})$; details on testing $H_0'(\xt{1},g)$ are given in Steps 3--5 of Algorithm \ref{alg:cluster} in Supplement \ref{app:cluster}. 
\end{enumerate}

Table \ref{tab:cells} presents a confusion matrix comparing the estimated clusters to the true cell subtypes; the clusters largely recover the true subtypes from $\xt{1}$, though Cluster 2 does partially blend the two types of T-cells. All pairwise selected null hypotheses are rejected, even after adjusting the $4(4-1)/2=6$ tests to control the family-wise error rate \citep{holm1979simple}. Together, these results imply that our procedure can detect latent cell subtypes and reject the null of no difference between identified groups.

\begin{table}
\centering
\begin{tabular}{l|rrrr}
\toprule
  & B-cells & Naive Cytotoxic T-cells & CD14 Monocytes & Regulatory T-cells\\
\midrule
Cluster 1 & 976 & 5 & 33 & 10\\
Cluster 2 & 22 & 982 & 11 & 518\\
Cluster 3 & 1 & 0 & 940 & 1\\
Cluster 4 & 0 & 11 & 16 & 468\\
\bottomrule
\end{tabular}
\caption{Confusion matrix comparing the estimated clusters resulting from the workflow described in Section \ref{sec:data} to the true cell subtypes provided by \cite{zheng2017massively} and \cite{duo2020systematic}.}
\label{tab:cells}
\end{table}

We conclude with Figure \ref{fig:scrna}, which provides a visual depiction of the analysis conducted in this section, providing intuition for the results. The left panel displays $\xt{1}_i/1_{50}^\top \xt{1}_i$ projected onto the first two principal components of $\xt{1}_i/1_{50}^\top \xt{1}_i$; estimated clusters are indicated by colour and true cell subtypes are indicated by shape. The panel visualises the agreement between estimated clusters and cell subtype annotations detailed in Table \ref{tab:cells}. The middle panel displays $\xt{2}_i/1_{50}^\top \xt{2}_i$ projected onto the first two principal components of $\xt{1}_i/1_{50}^\top \xt{1}_i$; as expected it is qualitatively similar to the left panel. The right panel displays $(\xt{2}_i-\hat\E[\Xt{2}_i|\Xt{1}_i=\xt{1}_i]/1_{50}^\top \xt{2}_i$ projected onto the first two principal components of $\xt{1}_i/1_{50}^\top \xt{1}_i$ and subset to Clusters 1 and 3. We see that after orthogonalizing, Clusters 1 and 3 remain separated, which leads to the rejection of the corresponding selected null hypothesis.

\begin{figure}[h]
\centering
\includegraphics[width=.9\linewidth]{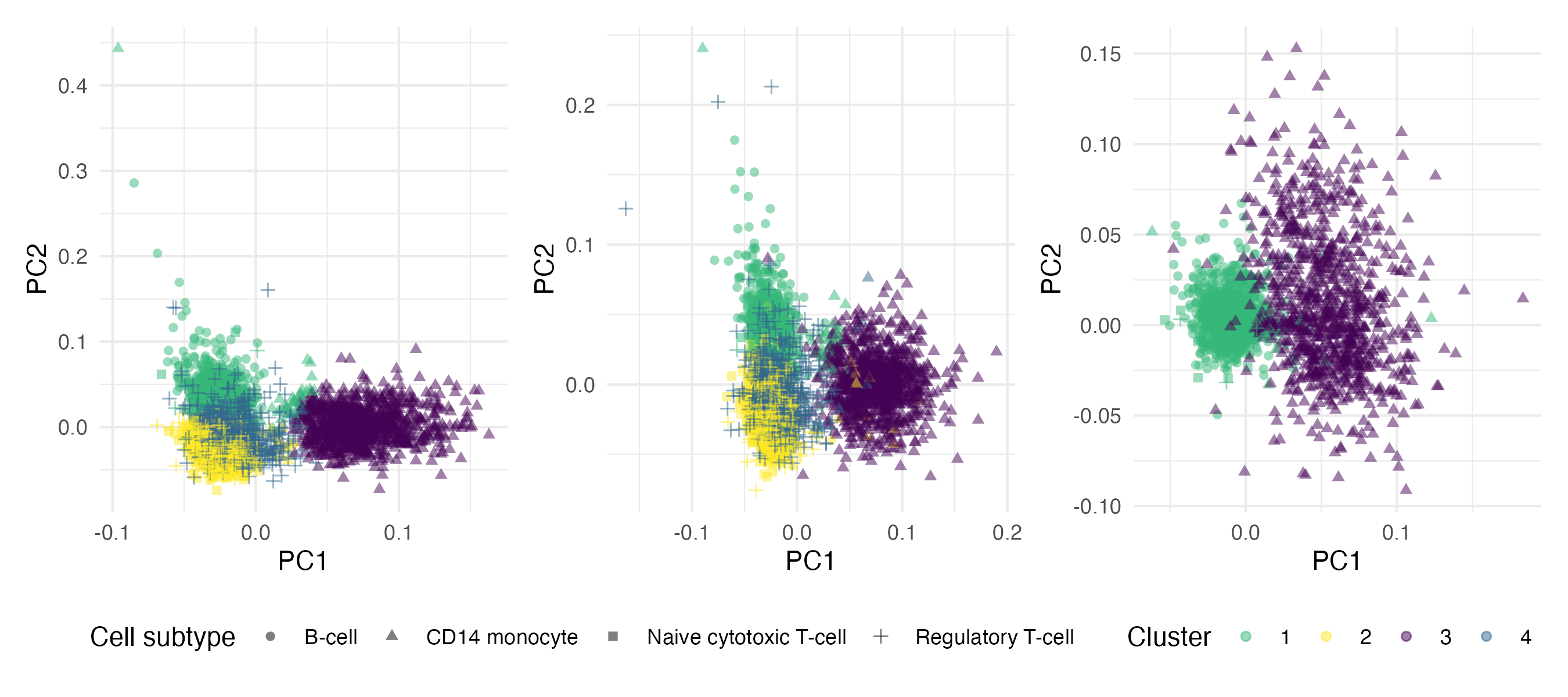}
\caption{Visualization of the orthogonalization procedure applied to the PBMC data in Section \ref{sec:data}. \textit{Left:} The normalized noisy gene expression profiles, $\xt{1}_i/1_{50}^\top \xt{1}_i$, projected onto the first two principal components of $\xt{1}_i/1_{50}^\top \xt{1}_i$. \textit{Middle:} The second set of normalized noisy gene expression profiles, $\xt{2}_i/1_{50}^\top \xt{2}_i$, projected onto the first two principal components of $\xt{1}_i/1_{50}^\top \xt{1}_i$. \textit{Right: } The second set of normalized noisy gene expression profiles after attempted orthogonalization, $(\xt{2}_i-\hat\E[\Xt{2}_i|\Xt{1}_i=\xt{1}_i])/1_{50}^\top \xt{2}_i$, subset to Clusters 1 and 3 and projected onto the first two principal components of $\xt{1}_i/1_{50}^\top \xt{1}_i$. 
In all panels, estimated clusters are indicated by colour and true cell subtypes are indicated by shape.}
\label{fig:scrna}
\end{figure}
\section{Discussion}
\label{sec:discussion}

In this paper, we propose a new framework for testing both pre-specified and data-driven hypotheses. Our approach is based on the observation that any null hypothesis can be recast as a particular orthogonality moment condition. For a pre-specified hypothesis, we propose the following procedure: 1. randomize $X$ with symmetric shift-family noise to construct $\Xt{1}$ and $\Xt{2}$; 2. orthogonalize $\Xt{2}$ against $\Xt{1}$ under the null; and 3. test whether orthogonalization was successful using a $\chi^2$ test. For the post-selection inference setting, we additionally select a hypothesis using $\Xt{1}$ after Step 1, then proceed analogously. The key to our proposal lies in our use of symmetric shift-family noise for the randomization step. This choice induces a universal form for the conditional mean of $\Xt{2}$ given $\Xt{1}$, which is easily computable in broad settings.

Our proposal is particularly powerful in the context of post-selection inference. Unlike much of the extant literature, our procedure is based on properties of conditional means, rather than on a complete characterization of a conditional distribution. By isolating exactly what feature of the conditional distribution is necessary for inference, we are able to drastically relax commonly-made assumptions. Unlike much of the randomization literature, we are not restricted to a small class of parametric families, and unlike conditional selective inference, we do not require an exact specification of the selection event. The flexibility of our proposal is illustrated in Sections \ref{sec:cases} and \ref{sec:data}, in which we conduct valid (post-selection) inference in diverse settings.

The orthogonalization procedure is modular by design; it can and should be adapted to the context at hand. For instance, in multivariate settings, knowledge about the dependence structure between the entries of $X_i$ can be used to simplify estimation of the conditional mean, $\Psi_n$ in Theorem \ref{thm:Cn}, and $\Omega_n$ in Theorem \ref{thm:Dn}. In a related setting, \cite{shah2020hardness} use the maximum studentized entry of a quantity resembling $C_n(\Xt{1},\Xt{2};\Theta_0)$ and $D_n(\xt{1},\xt{2};\Theta_0(\xt{1}))$ (defined in Sections \ref{sec:cmean} and \ref{sec:sel}, respectively) as their test statistic, arguing that the maximum behaves better than the $\ell_2$-norm. It would be interesting to clarify whether similar advantages can be realized in our case. 

The choice of test function $g$ drives the power of our method; it affects how much (if anything) is lost in the process of reformulating $H_0$ into $H_0'(g)$ (or $H_0(\xt{1})$ into $H_0'(\xt{1},g)$). We suspect that the optimal choice in practice will depend on the originating null hypothesis $H_0$. Our empirical studies suggest that more expressive test functions are to be preferred though this is not a universal truth: in Supplement \ref{app:cp} we find that an indicator for segment membership outperforms the data itself in the context of inference after changepoint detection. We leave to future work a systematic investigation of favourable properties of test functions.

Code to reproduce all results in this paper are available at \url{https://github.com/AmeerD/Orthogonalization}.

\if\blind0
\section*{Acknowledgement}
We thank Lucy L. Gao, Ellen Graham, Olivia McGough, Rui Wang, and Zichun Xu for helpful conversations that contributed to the direction of this project. We acknowledge funding from the following sources: 
Office of Naval Research,  National Science Foundation, and  National Institutes of Health of the United States to DW; and Natural Sciences and Engineering Research Council of Canada to AD.
\fi

\bibliographystyle{natbib}
\bibliography{dasi}

\newpage
\setcounter{figure}{0}
\setcounter{lemma}{0}
\setcounter{theorem}{0}
\setcounter{proposition}{0}
\setcounter{remark}{0}
\setcounter{corollary}{0}
\setcounter{definition}{0}
\setcounter{algorithm}{0}
\makeatletter
\renewcommand \thefigure{S\@arabic\c@figure}
\renewcommand \thelemma{S\@arabic\c@lemma}
\renewcommand{\theHlemma}{appendix.\@arabic\c@lemma}
\renewcommand \thetheorem{S\@arabic\c@theorem}
\renewcommand \theproposition{S\@arabic\c@proposition}
\renewcommand{\theHproposition}{appendix.\@arabic\c@proposition}
\renewcommand \theremark{S\@arabic\c@remark}
\renewcommand \thecorollary{S\@arabic\c@corollary}
\renewcommand \thedefinition{S\@arabic\c@definition}
\renewcommand \thealgorithm{S\@arabic\c@algorithm}
\renewcommand\theHalgorithm{S\@arabic\c@algorithm} 
\makeatother
\appendix
\section*{Supplementary Materials}
\section{Proofs of technical results}
\label{app:proofs}

\subsection{Proof of Theorem \ref{thm:t1e}}

\begin{proof}
By construction, $H_0\implies H_0'(g)$, which implies the result. 
\end{proof}

\subsection{Proof of Proposition \ref{prop:ratio}}

We begin with Lemma \ref{lem:cexp_generic}, which shows that the expectation of a function of $X_i$ conditional on $\Xt{1}_i$ admits a convenient form. 

\begin{lemma}
\label{lem:cexp_generic}
Consider $\Xt{1}_i$ and $\Xt{2}_i$ constructed from $X_i$ and $W_i$ using Algorithm \ref{alg:split}; thus, $W_i\overset{iid}{\sim} R(0,\Sigma)$ for a symmetric shift-family $R$. Then,
$$
\E_{\eta_i(\theta^*)}[s(X_i)|\Xt{1}_i=\xt{1}_i] = \frac{\int s(x_i)f_{X}(x_i;\eta_i(\theta))f_{R}(\xt{1}_i-x_i;0,\Sigma)\,dx_i}{\int f_{X}(x_i;\eta_i(\theta))f_{R}(\xt{1}_i-x_i;0,\Sigma)\,dx_i}=\frac{\E_{U_i}[s(U_i)f_{X}(U_i;\eta_i(\theta^*))]}{\E_{U_i}[f_{X}(U_i;\eta_i(\theta^*))]}.
$$
where $U_i\sim R(\xt{1}_i,\Sigma)$; i.e., it is a random variable that follows the distribution of the user-added noise $W_i$ shifted by $\xt{1}_i$.
\end{lemma}

\begin{proof}
For simplicity, assume that $X_i$ and $W_i$ are both continuous random variables. Similar arguments apply in the discrete case. Observe that
\begin{align*}
E_{\eta_i(\theta^*)}[s(X_i)|\Xt{1}_i=\xt{1}_i] &= E_{\eta_i(\theta^*)}[s(X_i)|X_i+W_i=\xt{1}_i] \\
&= \int s(x_i)f_{X|X+W}(x_i|x_i+w_i=\xt{1}_i;\eta_i(\theta^*))\,dx_i \\
&= \int s(x_i)\frac{f_{X,X+W}(x_i,x_i+w_i=\xt{1}_i;\eta_i(\theta^*))}{f_{X+W}(\xt{1}_i;\eta_i(\theta^*))}\,dx_i \\
&= \frac{\int s(x_i)f_{X}(x_i;\eta_i(\theta^*))f_{R}(\xt{1}_i-x_i;0,\Sigma)\,dx_i}{\int f_{X}(x_i;\eta_i(\theta^*))f_{R}(\xt{1}_i-x_i;0,\Sigma)\,dx_i} \\
&= \frac{\int s(x_i)f_{X}(x_i;\eta_i(\theta^*))f_{R}(x_i;\xt{1}_i,\Sigma)\,dx_i}{\int f_{X}(x_i;\eta_i(\theta^*))f_{R}(x_i;\xt{1}_i,\Sigma)\,dx_i} \\
&=\frac{\E_{U_i}[s(U_i)f_{X}(U_i;\eta_i(\theta^*))]}{\E_{U_i}[f_{X}(U_i;\eta_i(\theta^*))]}.
\end{align*}
\end{proof}

We now prove Proposition \ref{prop:ratio}.

\begin{proof}
Substituting the definitions of $\Xt{1}_i$ and $\Xt{2}_i$ into $E_{\eta_i(\theta)}[\Xt{2}_i|\Xt{1}_i=\xt{1}_i]$, then applying Lemma \ref{lem:cexp_generic} yields
\begin{align*}
E_{\eta_i(\theta)}[\Xt{2}_i|\Xt{1}_i=\xt{1}_i] &= E_{\eta_i(\theta)}[X_i-W_i|X_i+W_i=\xt{1}_i] \\
&= 2E_{\eta_i(\theta)}[X_i|X_i+W_i=\xt{1}_i]-\xt{1}_i \\
&= 2\frac{\int x_if_{X}(x_i;\eta_i(\theta))f_{R}(x_i;\xt{1}_i,\Sigma)\,dx_i}{\int f_{X}(x_i;\eta_i(\theta))f_{R}(x_i;\xt{1}_i,\Sigma)\,dx_i}-\xt{1}_i \\
&= 2\frac{N(\xt{1}_i)}{D(\xt{1}_i)}-\xt{1}_i
\end{align*}
where in the last step we define 
$$
N(\xt{1}_i) = \int x_if_{X}(x_i;\eta_i(\theta))f_{R}(x_i;\xt{1}_i,\Sigma)\,dx_i \quad\text{and}\quad D(\xt{1}_i) = \int f_{X}(x_i;\eta_i(\theta))f_{R}(x_i;\xt{1}_i,\Sigma)\,dx_i.
$$

Finally, \eqref{eq:Uratio} can be recovered by rewriting the above as expectations taken with respect to a random variable $U_i\sim R(\xt{1}_i,\Sigma)$, and \eqref{eq:Xratio} can be recovered by rewriting the above as expectations taken with respect to $X_i$.
\end{proof}

\subsection{Proof of Lemma \ref{lem:ASL}}

\begin{proof}
To ease notation, let $P_n(f)=\frac{1}{n}\sum_{i=1}^n f(\xt{1}_i,\xt{2}_i)$ denote the sample mean of some function $f$ and let $P(f)=\int f(\xt{1},\xt{2})\,dP(\xt{1},\xt{2})$ denote the corresponding population mean. Assume without loss of generality that $|M_k|=n/K$ for $k=1,\dots,K$ and let $P_{n,k}(f)=\frac{K}{n}\sum_{i\in M_k}f(\xt{1}_i,\xt{2}_i)$ denote the sample mean within fold $k$. Then, $C_n(\xt{1},\xt{2};\Theta_0)$ can be expanded as follows:
\begin{align*}
&C_n(\xt{1},\xt{2};\Theta_0) \\
&= P_n\left[\left(\Xt{1}+\Xt{2}-2\frac{\hat N_n^{(-k)}}{\hat D_n^{(-k)}}\right)g^\top\right] \\ 
&= P_n\left[\left(\Xt{1}+\Xt{2}-2\frac{N}{D}\right)g^\top\right] - 2P_n\left[\left(\frac{\hat N_n^{(-k)}}{\hat D_n^{(-k)}}-\frac{N}{D}\right)g^\top\right] \\ 
&= P_n\left[\left(\Xt{1}+\Xt{2}-2\frac{N}{D}\right)g^\top\right] - \frac{2}{K}\sum_{k=1}^KP_{n,k}\left[\left(\frac{\hat N_n^{(-k)}}{\hat D_n^{(-k)}}-\frac{N}{D}\right)g^\top\right] \\ 
&= P_n\left[\left(\Xt{1}+\Xt{2}-2\frac{N}{D}\right)g^\top\right] - \frac{2}{K}\sum_{k=1}^KP\left[\left(\frac{\hat N_n^{(-k)}}{\hat D_n^{(-k)}}-\frac{N}{D}\right)g^\top\right]+ \frac{2}{K}\sum_{k=1}^K(P-P_{n,k})\left[\left(\frac{\hat N_n^{(-k)}}{\hat D_n^{(-k)}}-\frac{N}{D}\right)g^\top\right].
\end{align*}

It remains to show that each of the three terms in the above expansion are either asymptotically linear or negligible. The first term is linear by definition. The second and third term require further study. First, observe that since $\hat N_n^{(-k)}(\cdot)$ and $\hat D_n^{(-k)}(\cdot)$ are asymptotically linear estimators of $N(\cdot)$ and $D(\cdot)$, respectively, and that $D(\cdot)>0$ by construction (the realized value of $x_i$ used to form $\xt{1}_i$ will by symmetry always be in the support of $R(\xt{1}_i,\Sigma)$), an application of the delta method yields that 
\begin{align*}
\frac{\hat N_n^{(-k)}(\cdot)}{\hat D_n^{(-k)}(\cdot)} - \frac{N(\cdot)}{D(\cdot)} &= \frac{1}{n-|M_k|}\sum_{i\not\in M_k} \left\{\varphi_N(\xt{1}_i,\xt{2}_i;\cdot)-\frac{N(\cdot)}{D(\cdot)}\varphi_D(\xt{1}_i,\xt{2}_i;\cdot)\right\}\frac{1}{D(\cdot)} + o_P(n^{-1/2}) \\
&= \frac{1}{n-|M_k|}\sum_{i\not\in M_k} \varphi_{ND}(\xt{1}_i,\xt{2}_i;\cdot) + o_P(n^{-1/2}).
\end{align*}
That is, $\frac{\hat N_n^{(-k)}(\cdot)}{\hat D_n^{(-k)}(\cdot)}$ is an asymptotically linear estimator of $\frac{N(\cdot)}{D(\cdot)}$ with influence function $\varphi_{ND}(\xt{1}_i,\xt{2}_i;\cdot)$. 

The second term simplifies as
\begin{align*}
\frac{2}{K}\sum_{k=1}^KP\left[\left(\frac{\hat N_n^{(-k)}}{\hat D_n^{(-k)}}-\frac{N}{D}\right)g^\top\right] &= \frac{2}{K}\sum_{k=1}^K P\left[\frac{1}{n-n/K}\sum_{i\not\in M_k} \varphi_{ND}(\xt{1}_i,\xt{2}_i;\xt{1})g(\xt{1})^\top + o_P(n^{-1/2})\right] \\
&= \frac{2}{n(K-1)}\sum_{k=1}^K \sum_{i\not\in M_k}P\left[\varphi_{ND}(\xt{1}_i,\xt{2}_i;\xt{1})g(\xt{1})^\top + o_P(n^{-1/2})\right] \\
&= \frac{2}{n(K-1)}\sum_{i=1}^n (K-1)P\left[\varphi_{ND}(\xt{1}_i,\xt{2}_i;\xt{1})g(\xt{1})^\top + o_P(n^{-1/2})\right] \\
&= \frac{1}{n}\sum_{i=1}^n 2\int\varphi_{ND}(\xt{1}_i,\xt{2}_i;\xt{1})g(\xt{1})^\top\,dP(\xt{1}) + o_P(n^{-1/2}).
\end{align*}
This shows that the second term is asymptotically linear as well.

Finally, the third term is a standard empirical process term and is negligible as a result of our use of cross-fitting in the construction of $\hat N_n^{(-k)}(\cdot)$ and $\hat D_n^{(-k)}(\cdot)$ and the restriction that $\hat N_n^{(-k)}(\cdot)$ and $\hat D_n^{(-k)}(\cdot)$ are asymptotically linear \citep{chernozhukov2022locally, kennedy2024semiparametric}.

Combining the above, it follows that 
\small
$$
C_n(\xt{1},\xt{2};\Theta_0) = \frac{1}{n}\sum_{i=1}^n\left[\left(\xt{1}_i+\xt{2}_i-2\frac{N(\xt{1}_i)}{D(\xt{1}_i)}\right)g(\xt{1}_i)^\top -2\int\varphi_{ND}(\xt{1}_i,\xt{2}_i;\xt{1})g(\xt{1})^\top dP(\xt{1})\right] + o_P(n^{-1/2}),
$$
\normalsize
thus proving the result.
\end{proof}

\subsection{Proof of Theorem \ref{thm:Cn}}

\begin{proof}
The result follows immediately by combining Lemma \ref{lem:ASL} with an application of the Lindeberg-Feller central limit theorem.
\end{proof}

\subsection{Proof of Corollary \ref{cor:Tn}}

\begin{proof}
The result follows from Slutsky's theorem and the continuous mapping theorem.
\end{proof}

\subsection{Proof of Proposition \ref{prop:asym}}

\begin{proof}
The result is a standard fact about Wald tests \citep{lehmann2005testing}. It can be seen as a consequence of the Portmanteau lemma.
\end{proof}

\subsection{Proof of Proposition \ref{prop:delta}}

\begin{proof}
For $\hat N_n^{(-k)}(\xt{1}_{i'})$, observe that 
\begin{align*}
&\hat N_n^{(-k)}(\xt{1}_{i'}) - N(\xt{1}_{i'}) \\
&= \E_{U_{i'}}[U_{i'}f_{X}(U_{i'};\eta_{i'}(\hat\theta_n^{(-k)}))] - \E_{U_{i'}}[U_{i'}f_{X}(U_{i'};\eta_{i'}(\theta^*))] \\
&= \nabla_\theta \left(\E_{U_{i'}}[U_{i'}f_{X}(U_{i'};\eta_{i'}(\theta))]\right)\Big\vert_{\theta=\theta^*}\left(\hat\theta_n^{(-k)} - \theta^*\right) + o_P(n^{-1/2}) \\
&= \E_{U_{i'}}[U_{i'}\nabla_\theta f_X(x;\eta_{i'}(\theta))^\top \mid_{\theta=\theta^*}]\left(\frac{1}{n-|M_k|}\sum_{i\not\in M_k}\varphi_\theta(x_i) + o_P(n^{-1/2})\right) + o_P(n^{-1/2}) \\
&=\frac{1}{n-|M_k|}\sum_{i\not\in M_k}\E_{U_{i'}}[U_{i'}\nabla_\theta f_X(x;\eta_{i'}(\theta))^\top \mid_{\theta=\theta^*}] \varphi_\theta(x_i) + o_P(n^{-1/2}),
\end{align*}
where the second equality follows from a Taylor expansion of $\E_{U_{i'}}[U_{i'}f_{X}(U_{i'};\eta_{i'}(\hat\theta_n^{(-k)}))]$ around $\theta^*$.

The result for $\hat D_n^{(-k)}(\xt{1}_{i'})$ follows similarly.
\end{proof}

\subsection{Proof of Corollary \ref{cor:delta}}

\begin{proof}
Starting from \eqref{eq:varphi-C}, substituting in the definition of $\varphi_{ND}(\cdot)$, followed by the definitions of $\varphi_N(\cdot)$ and $\varphi_D(\cdot)$ from Proposition \ref{prop:delta} yields
\small
\begin{align*}
&\varphi_C(\xt{1}_i,\xt{2}_i) \\
&= \left(\xt{1}_i+\xt{2}_i-2\frac{N(\xt{1}_i)}{D(\xt{1}_i)}\right)g(\xt{1}_i) -2\int\varphi_{ND}(\xt{1}_i,\xt{2}_i;s)g(s) dP_{\xt{1}}(s) \\
&= \left(\xt{1}_i+\xt{2}_i-2\frac{N(\xt{1}_i)}{D(\xt{1}_i)}\right)g(\xt{1}_i) -2\int\left\{\varphi_N(\xt{1}_i,\xt{2}_i;s)-\frac{N(s)}{D(s)}\varphi_D(\xt{1}_i,\xt{2}_i;s)\right\}\frac{g(s)}{D(s)} dP_{\xt{1}}(s) \\
&= \left(\xt{1}_i+\xt{2}_i-2\frac{N(\xt{1}_i)}{D(\xt{1}_i)}\right)g(\xt{1}_i) \\
&\quad-2\int\left\{\E_{U_{s}}[U_{s}\nabla_\theta f_X(x;\eta_{s}(\theta))^\top \mid_{\theta=\theta^*}] \varphi_\theta(x_i)-\frac{N(s)}{D(s)}\E_{U_{s}}[\nabla_\theta f_X(x;\eta_{s}(\theta))^\top \mid_{\theta=\theta^*}] \varphi_\theta(x_i)\right\}\frac{g(s)}{D(s)} dP_{\xt{1}}(s) \\
&= \left(\xt{1}_i+\xt{2}_i-2\frac{N(\xt{1}_i)}{D(\xt{1}_i)}\right)g(\xt{1}_i) \\
&\quad-2\int\left\{\E_{U_{s}}[U_{s}\nabla_\theta f_X(x;\eta_{s}(\theta))^\top \mid_{\theta=\theta^*}]-\frac{N(s)}{D(s)}\E_{U_{s}}[\nabla_\theta f_X(x;\eta_{s}(\theta))^\top \mid_{\theta=\theta^*}]\right\}\frac{g(s)}{D(s)} dP_{\xt{1}}(s)\varphi_\theta(x_i) \\
&= \left(\xt{1}_i+\xt{2}_i-2\frac{N(\xt{1}_i)}{D(\xt{1}_i)}\right)g(\xt{1}_i) - 2A(g)\varphi_\theta(x_i).
\end{align*}
\normalsize
\end{proof}

\subsection{Proof of Proposition \ref{prop:power}}

\begin{proof}
The asymptotic normality of $\vect{C_n(\Xt{1},\Xt{2};\Theta_0)}$ follows from an application of the Lindeberg-Feller central limit theorem. It remains to clarify the form of $\E\left[C_n(\Xt{1},\Xt{2};\Theta_0)\right]$. Observe that
\small
\begin{align*}
&\E\left[C_n(\Xt{1},\Xt{2};\theta_0)\right]\\
&=\frac{1}{n}\sum_{i=1}^n\E_{\eta_i(\theta_1)}\left[\left(\Xt{2}_i-h(\Xt{1}_i,\theta_0)\right)g(\Xt{1}_i)^\top\right] \\
&=\frac{1}{n}\sum_{i=1}^n\E_{\eta_i(\theta_1)}\left[\left(\Xt{2}_i-\E_{\eta_i(\theta_0)}[\Xt{2}_i|\Xt{1}_i=\xt{1}_i]\right)g(\Xt{1}_i)^\top\right] \\
&=\frac{1}{n}\sum_{i=1}^n\E_{\eta_i(\theta_1)}\left[\left(\Xt{2}_i-\E_{\eta_i(\theta_1)}[\Xt{2}_i|\Xt{1}_i=\xt{1}_i]\right)g(\Xt{1}_i)^\top\right. \\
&\quad\left.+\left(\E_{\eta_i(\theta_1)}[\Xt{2}_i|\Xt{1}_i=\xt{1}_i]-\E_{\eta_i(\theta_0)}[\Xt{2}_i|\Xt{1}_i=\xt{1}_i]\right)g(\Xt{1}_i)^\top\right] \\
&=\frac{1}{n}\sum_{i=1}^n\E_{\eta_i(\theta_1)}\left[\left(\E_{\eta_i(\theta_1)}[\Xt{2}_i|\Xt{1}_i=\xt{1}_i]-\E_{\eta_i(\theta_0)}[\Xt{2}_i|\Xt{1}_i=\xt{1}_i]\right)g(\Xt{1}_i)^\top\right] \\
&=\frac{2}{n}\sum_{i=1}^n\E_{\eta_i(\theta_1)}\left[\left(\frac{\E_{U_i}[U_if_{X}(U_i;\eta_i(\theta_1))]}{\E_{U_i}[f_{X}(U_i;\eta_i(\theta_1))]}-\frac{\E_{U_i}[U_if_{X}(U_i;\eta_i(\theta_0))]}{\E_{U_i}[f_{X}(U_i;\eta_i(\theta_0))]}\right)g(\Xt{1}_i)^\top\right],
\end{align*}
\normalsize
where the last equality follows from Proposition \ref{prop:ratio}.
\end{proof}

\subsection{Proof of Proposition \ref{prop:cIF}}

\begin{proof}
Taking the conditional expectation of the simplification of $\varphi_C(\cdot)$ given in Corollary \ref{cor:delta} yields
\begin{align*}
&\E_{\eta_i(\theta^*)}\left[\varphi_C(\Xt{1}_i,\Xt{2}_i)|\Xt{1}\right] \\
&= \E_{\eta_i(\theta^*)}\left[\left(\Xt{1}_i+\Xt{2}_i-2\frac{N(\Xt{1}_i)}{D(\Xt{1}_i)}\right)g(\Xt{1}_i) -2A(g)\varphi_\theta(X_i;\theta^*)\Big\vert\Xt{1}\right] \\
&= \E_{\eta_i(\theta^*)}\left[\left(\Xt{2}_i-\E\left[\Xt{2}_i|\Xt{1}_i\right]\right)g(\Xt{1}_i) -2A(g)\varphi_\theta(X_i;\theta^*)\Big\vert\Xt{1}\right] \\
&= -2A(g)\E_{\eta_i(\theta^*)}\left[\varphi_\theta(X_i;\theta^*)\Big\vert\Xt{1}_i\right] \\
&= -2A(g)\frac{\E_{U_i}[\varphi_\theta(U_i;\theta^*)f_{X}(U_i;\eta_i(\theta^*))]}{\E_{U_i}[f_{X}(U_i;\eta_i(\theta^*))]},
\end{align*}
where the first equality follows from Corollary \ref{cor:delta}, the second follows from Proposition \ref{prop:ratio}, the third from properties of the expectation, and the fourth from applying Lemma \ref{lem:cexp_generic} with $s(X_i)=\varphi_\theta(X_i;\theta^*)$.
\end{proof}

\subsection{Proof of Lemma \ref{lem:DnIF}}
\label{app:lem2}

\begin{proof}
To ease notation, let $b(\xt{1}_i;\theta)=\frac{\E_{U_i}[\varphi_\theta(U_i;\theta)f_{X}(U_i;\eta_i(\theta))]}{\E_{U_i}[f_{X}(U_i;\eta_i(\theta))]}$ where $U_i\sim R(\xt{1}_i,\Sigma)$. It follows that we can write $D_n(\xt{1},\xt{2};\Theta_0(\xt{1}))$ in \eqref{eq:Dn} as
$$
D_n(\xt{1},\xt{2};\Theta_0(\xt{1})) = C_n(\Xt{1},\Xt{2};\Theta_0(\Xt{1})) + 2A(g)(I+B)^{-1}\frac{1}{n}\sum_{i=1}^n b(\Xt{1}_i;\hat\theta_n^{(-k)}).
$$

Recall from Corollary \ref{cor:delta} that $C_n(\xt{1},\xt{2};\Theta_0(\xt{1}))=P_n\left[\varphi_C(\xt{1},\xt{2})\right]+o_P(n^{-1/2})$. We begin by expanding $D_n(\xt{1},\xt{2};\Theta_0(\xt{1}))$ as follows:
\begin{align*}
&D_n(\xt{1},\xt{2};\Theta_0(\xt{1})) \\
&= P_n\left[\varphi_C(\Xt{1},\Xt{2})\right] + 2A(g)(I+B)^{-1}P_n\left[b(\Xt{1};\hat\theta_n^{(-k)})\right]  + o_P(n^{-1/2}) \\
&= P_n\left[\varphi_C(\Xt{1},\Xt{2})\right] + 2A(g)(I+B)^{-1}\frac{1}{K}\sum_{k=1}^KP_{n,k}\left[b(\Xt{1};\hat\theta_n^{(-k)})\right]  + o_P(n^{-1/2}) \\
&= P_n\left[\varphi_C(\Xt{1},\Xt{2})\right] + 2A(g)(I+B)^{-1}\frac{1}{K}\sum_{k=1}^K\left(P\left[b(\Xt{1};\hat\theta_n^{(-k)})\right] + (P_{n,k}-P)\left[b(\Xt{1};\hat\theta_n^{(-k)})\right]\right)  + o_P(n^{-1/2}) \\
&= P_n\left[\varphi_C(\Xt{1},\Xt{2})\right] + 2A(g)(I+B)^{-1}\frac{1}{K}\sum_{k=1}^K\left(P\left[b(\Xt{1};\hat\theta_n^{(-k)})\right] + (P_{n,k}-P)\left[b(\Xt{1};\theta^*)\right]\right. \\
&\quad\left.+ (P_{n,k}-P)\left[b(\Xt{1};\hat\theta_n^{(-k)})-b(\Xt{1};\theta^*)\right]\right)  + o_P(n^{-1/2}) \\
&= P_n\left[\varphi_C(\Xt{1},\Xt{2})\right] + 2A(g)(I+B)^{-1}\frac{1}{K}\sum_{k=1}^K\left(P\left[b(\Xt{1};\hat\theta_n^{(-k)})\right] + P_{n,k}\left[b(\Xt{1};\theta^*)\right]\right)  + o_P(n^{-1/2}),
\end{align*}
where in the last step we recall from Proposition \ref{prop:cIF} that $b(\xt{1}_i;\theta)=\frac{\E_{U_i}[\varphi_\theta(U_i;\theta)f_{X}(U_i;\eta_i(\theta))]}{\E_{U_i}[f_{X}(U_i;\eta_i(\theta))]}=\E_{\eta_i(\theta)}\left[\varphi_\theta(X_i;\theta)|\Xt{1}_i\right]$, which implies that $P\left[b(\Xt{1};\theta^*)\right]=\E_{\eta_i(\theta^*)}\left[\E_{\eta_i(\theta^*)}\left[\varphi_\theta(X_i;\theta^*)|\Xt{1}_i\right]\right]=0$, and then we note that the empirical process term is negligible due to cross-fitting \citep{chernozhukov2022locally, kennedy2024semiparametric}.

We focus our attention on $P\left[b(\Xt{1};\hat\theta_n^{(-k)})\right]$ as the remaining terms are already linear. Without loss of generality, assume that the cross-fitting folds $M_1,\dots,M_K$ satisfy $|M_k|=n/K$ for $k=1,\dots,K$, and observe that a Taylor expansion around $\theta^*$ yields
\begin{align*}
P\left[b(\Xt{1};\hat\theta_n^{(-k)})\right]&=P\left[b(\Xt{1};\theta^*) + \left(\nabla_\theta b(\Xt{1};\theta)\Big{\vert}_{\theta=\theta^*}\right)^\top\left(\hat\theta_n^{(-k)}-\theta^*\right)+o_P(n^{-1/2})\right] \\
&=P\left[\left(\nabla_\theta b(\Xt{1};\theta)\Big{\vert}_{\theta=\theta^*}\right)^\top\frac{1}{n-n/K}\sum_{i\not\in M_k}\varphi_\theta(x_i;\theta^*)\right] +o_P(n^{-1/2}) \\
&=\frac{K}{n(K-1)}\sum_{i\not\in M_k}P\left[\left(\nabla_\theta b(\Xt{1};\theta)\Big{\vert}_{\theta=\theta^*}\right)^\top\right]\varphi_\theta(x_i;\theta^*) +o_P(n^{-1/2}) \\
&=\frac{K}{n(K-1)}\sum_{i\not\in M_k}B\varphi_\theta(x_i;\theta^*) +o_P(n^{-1/2}).
\end{align*}
The first equality follows from the Taylor expansion, the second from the fact that $\hat\theta_n^{(-k)}$ is an asymptotically linear estimator of $\theta^*$ with influence function $\varphi_\theta(\cdot)$ as well as the fact that $P[b(\Xt{1};\theta^*)]=0$, the third is a rearrangement of terms, the fourth from the definition of $B$ below \eqref{eq:Dn}. Combining all $k$ folds yields
\begin{equation}
\label{eq:Pb}
\frac{1}{K}\sum_{k=1}^KP\left[b(\Xt{1};\hat\theta_n^{(-k)})\right]=\frac{1}{n}\sum_{i=1}^nB\varphi_\theta(x_i;\theta^*) +o_P(n^{-1/2}) = BP_n\left[\varphi_\theta(X;\theta^*)\right] +o_P(n^{-1/2})
\end{equation}

Returning to $D_n(\xt{1},\xt{2};\Theta_0(\xt{1}))$, we have that
\small
\begin{align*}
&D_n(\xt{1},\xt{2};\Theta_0(\xt{1})) \\
&= P_n\left[\varphi_C(\Xt{1},\Xt{2})\right] + 2A(g)(I+B)^{-1}\frac{1}{K}\sum_{k=1}^K\left(P\left[b(\Xt{1};\hat\theta_n^{(-k)})\right] + P_{n,k}\left[b(\Xt{1};\theta^*)\right]\right)  + o_P(n^{-1/2}) \\
&= P_n\left[\varphi_C(\Xt{1},\Xt{2})\right] + 2A(g)(I+B)^{-1}\left(BP_n\left[\varphi_\theta(X;\theta^*)\right] + P_n\left[b(\Xt{1};\theta^*)\right]\right)  + o_P(n^{-1/2}) \\
&= P_n\left[\varphi_C(\Xt{1},\Xt{2}) + 2A(g)(I+B)^{-1}\left(B\varphi_\theta(X;\theta^*) + b(\Xt{1};\theta^*)\right)\right]  + o_P(n^{-1/2}) \\
&= P_n\left[\left(\Xt{1}+\Xt{2}-2\frac{N}{D}\right)g -2A(g)\varphi_\theta(X;\theta^*) + 2A(g)(I+B)^{-1}\left(B\varphi_\theta(X;\theta^*) + b(\Xt{1};\theta^*)\right)\right]  + o_P(n^{-1/2}) \\
&= P_n\left[\left(\Xt{1}+\Xt{2}-2\frac{N}{D}\right)g -2A(g)\left(\varphi_\theta(X;\theta^*) -(I+B)^{-1}\left(B\varphi_\theta(X;\theta^*) + b(\Xt{1};\theta^*)\right)\right)\right]  + o_P(n^{-1/2}) \\
&= P_n\left[\left(\Xt{1}+\Xt{2}-2\frac{N}{D}\right)g -2A(g)\left(\left(I -(I+B)^{-1}B\right)\varphi_\theta(X;\theta^*) - (I+B)^{-1}b(\Xt{1};\theta^*)\right)\right]  + o_P(n^{-1/2}) \\
&= P_n\left[\left(\Xt{1}+\Xt{2}-2\frac{N}{D}\right)g -2A(g)\left((I+B)^{-1}\varphi_\theta(X;\theta^*) - (I+B)^{-1}b(\Xt{1};\theta^*)\right)\right]  + o_P(n^{-1/2}) \\
&= P_n\left[\left(\Xt{1}+\Xt{2}-2\frac{N}{D}\right)g -2A(g)(I+B)^{-1}\left(\varphi_\theta(X;\theta^*) - b(\Xt{1};\theta^*)\right)\right]  + o_P(n^{-1/2}),
\end{align*}
\normalsize
where the first equality restates the expression for $D_n(\xt{1},\xt{2};\Theta_0(\xt{1}))$, the second substitutes in \eqref{eq:Pb}, the third gathers all terms inside the $P_n$, the fourth substitutes the definition of $\varphi_C(\cdot)$ from \eqref{eq:varphi-C}, the fifth factors out $2A(g)$, the sixth gathers like terms, the seventh uses that $I-(I+B)^{-1}B=(I+B)^{-1}$, and the eighth factors out $(I+B)^{-1}$.

The result follows after we recall from Proposition \ref{prop:cIF} that $b(\Xt{1};\theta^*)=\E_{\eta_i(\theta^*)}\left[\varphi_\theta(X_i;\theta^*)|\Xt{1}_i\right]$.
\end{proof}

\subsection{Proof of Theorem \ref{thm:Dn}}

\begin{proof}
To prove the result, it suffices to show that
\begin{equation}
\label{eq:intermediate}
P\left(\sqrt{n}\Omega_n^{-1/2}\vect{\frac{1}{n}\sum_{i=1}^n \varphi_D(\xt{1}_i,\xt{2}_i)} \le x \Big\vert\mathcal{F}_n \right)\overset{p}\to \Phi_p(x),
\end{equation}
where $\varphi_D(\cdot)$ is defined in Lemma \ref{lem:DnIF}. Since $D_n(\Xt{1},\Xt{2};\Theta_0(\Xt{1}))-\frac{1}{n}\sum_{i=1}^n \varphi_D(\xt{1}_i,\xt{2}_i)=o_P(n^{-1/2})$, the result then follows from an application of the conditional Slutsky's theorem \citep{niu2024reconciling}.

To prove \eqref{eq:intermediate}, we apply a conditional Lindeberg-Feller central limit theorem \citep{bulinski2017conditional, zhao2025imputation}. In addition to the conditional Lindeberg-Feller condition which we assume in the statement of the theorem, this requires that $\E\left[\varphi_D(\Xt{1}_i,\Xt{2}_i)|\mathcal{F}_n\right]=0$ and that $\varphi_D(\Xt{1}_i,\Xt{2}_i)$ are independent conditional on $\mathcal{F}_n$. To establish the former, observe that
\begin{align*}
&\E\left[\varphi_D(\Xt{1}_i,\Xt{2}_i)|\mathcal{F}_n\right] \\
&= \E\left[\left(\Xt{1}_i+\Xt{2}_i-2\frac{N(\Xt{1}_i)}{D(\Xt{1}_i)}\right)g(\Xt{1}_i) -2A(g)(I+B)^{-1}\left(\varphi_\theta(X_i;\theta^*)-\E_{\eta_i(\theta^*)}\left[\varphi_\theta(X_i;\theta^*)|\Xt{1}_i\right]\right)\Bigg\vert\mathcal{F}_n\right] \\
&= \E\left[\Xt{2}_i-\E\left[\Xt{2}_i|\Xt{1}_i\right]\Bigg\vert\mathcal{F}_n\right]g(\Xt{1}_i) -2A(g)(I+B)^{-1}\E\left[\varphi_\theta(X_i;\theta^*)-\E_{\eta_i(\theta^*)}\left[\varphi_\theta(X_i;\theta^*)|\Xt{1}_i\right]\Bigg\vert\mathcal{F}_n\right] \\
&=0.
\end{align*}
Here, the first equality results from the definition of $\varphi_D(\cdot)$ in Lemma \ref{lem:DnIF}, the second follows from Proposition \ref{prop:ratio}, and the third equality follows from the tower rule.

Finally, note that $\varphi_D(\Xt{1}_i,\Xt{2}_i)$ are independent conditional on $\mathcal{F}_n$, as conditioning on $\mathcal{F}_n$ eliminates any dependence induced by selecting $H_0(\xt{1})$.
\end{proof}
\section{Additional technical details}

\subsection{Recovering $(I+B)^{-1}$ through iterative debiasing}
\label{app:geometric}

In Section \ref{sec:sel},  $D_n(\xt{1},\xt{2};\Theta_0(\xt{1}))$ in \eqref{eq:Dn} is a conditionally debiased variant of $C_n(\xt{1},\xt{2};\Theta_0(\xt{1}))$ that involves a plug-in debiasing term rescaled by the constant $(I+B)^{-1}$. To provide insight into $(I+B)^{-1}$, here we show that $D_n(\xt{1},\xt{2};\Theta_0(\xt{1}))$ can be recovered through an iterative sequence of plug-in debiasing steps that begins with $C_n^{(1)}(\Xt{1},\Xt{2};\Theta_0(\Xt{1}))$ in \eqref{eq:Cn1}, the naive plug-in estimator of $C_n'(\Xt{1},\Xt{2};\Theta_0(\Xt{1}))$ in \eqref{eq:Cn'} with the form of the plug-in term given by Proposition \ref{prop:cIF}. As in the proof of Lemma \ref{lem:DnIF}, to ease notation we define $b(\xt{1}_i;\theta)=\frac{\E_{U_i}[\varphi_\theta(U_i;\theta)f_{X}(U_i;\eta_i(\theta))]}{\E_{U_i}[f_{X}(U_i;\eta_i(\theta))]}$ where $U_i\sim R(\xt{1}_i,\Sigma)$.

Observe that $C_n^{(1)}(\Xt{1},\Xt{2};\Theta_0(\Xt{1}))$ in \eqref{eq:Cn1} can be written as
\begin{align*}
&C_n^{(1)}(\Xt{1},\Xt{2};\Theta_0(\Xt{1})) \\
&= C_n(\Xt{1},\Xt{2};\Theta_0(\Xt{1})) + 2A(g)\frac{1}{n}\sum_{i=1}^n b(\Xt{1}_i;\hat\theta_n^{(-k)}) \\
&= P_n\left[\varphi_C(\Xt{1},\Xt{2})\right]+2A(g)P_n\left[b(\Xt{1};\hat\theta_n^{(-k)})\right] + o_P(n^{-1/2}),
\end{align*}
where we recall the definition of $\varphi_C$ in \eqref{eq:varphi-C}.

Focusing on $P_n\left[b(\Xt{1};\hat\theta_n^{(-k)})\right]$, observe that 
\begin{align}
\begin{split}
\label{eq:Pb2}
&P_n\left[b(\Xt{1};\hat\theta_n^{(-k)})\right] \\
&=P_n\left[b(\Xt{1};\theta^*)\right] + P_n\left[b(\Xt{1};\hat\theta_n^{(-k)})-b(\Xt{1};\theta^*)\right] \\
&=P_n\left[b(\Xt{1};\theta^*)\right] +  \frac{1}{K}\sum_{k=1}^KP_{n,k}\left[b(\Xt{1};\hat\theta_n^{(-k)})-b(\Xt{1};\theta^*)\right] \\
&=P_n\left[b(\Xt{1};\theta^*)\right] +  \frac{1}{K}\sum_{k=1}^K\left(P\left[b(\Xt{1};\hat\theta_n^{(-k)})-b(\Xt{1};\theta^*)\right] + (P_{n,k}-P)\left[b(\Xt{1};\hat\theta_n^{(-k)})-b(\Xt{1};\theta^*)\right]\right) \\
&=P_n\left[b(\Xt{1};\theta^*)\right] +  \frac{1}{K}\sum_{k=1}^KP\left[b(\Xt{1};\hat\theta_n^{(-k)})\right] + o_P(n^{-1/2}) \\
&=P_n\left[b(\Xt{1};\theta^*)\right] +  BP_n\left[\varphi_\theta(X;\theta^*)\right] +o_P(n^{-1/2}),
\end{split}
\end{align}
where the second-to-last step uses the fact that $P\left[b(\Xt{1};\theta^*)\right]=0$ (see the Proof of Lemma \ref{lem:DnIF} in Supplement \ref{app:lem2}) and that the empirical process term is negligible by cross-fitting, and the final step uses \eqref{eq:Pb}. Putting the pieces together, we have that 
$$
C_n^{(1)}(\Xt{1},\Xt{2};\Theta_0(\Xt{1})) = P_n\left[\varphi_C(\Xt{1},\Xt{2})+2A(g)b(\Xt{1};\theta^*)\right] + 2A(g)BP_n\left[\varphi_\theta(X;\theta^*)\right]  + o_P(n^{-1/2}).
$$

The above clarifies the issue with $C_n^{(1)}(\Xt{1},\Xt{2};\Theta_0(\Xt{1}))$: while the first term on the right-hand side is the conditionally mean-zero linear component of $C_n'(\Xt{1},\Xt{2};\Theta_0(\Xt{1}))$ in \eqref{eq:Cn'}, there is an extra term $2A(g)BP_n\left[\varphi_\theta(X;\theta^*)\right]$. This new term is not conditionally mean-zero; in fact, it follows from Proposition \ref{prop:cIF} that its conditional mean is $2A(g)BP_n\left[b(\Xt{1};\theta^*)\right]$. Unfortunately, we cannot debias $C_n^{(1)}(\Xt{1},\Xt{2};\Theta_0(\Xt{1}))$ by subtracting off $2A(g)BP_n\left[b(\Xt{1};\theta^*)\right]$, as it depends on $\theta^*$. Instead, consider the following, which subtracts off a plug-in correction:
$$
C_n^{(2)}(\Xt{1},\Xt{2};\Theta_0(\Xt{1})) = C_n^{(1)}(\Xt{1},\Xt{2};\Theta_0(\Xt{1})) - 2A(g)BP_n\left[b(\Xt{1};\hat\theta_n^{(-k)})\right].
$$

As with $C_n^{(1)}(\Xt{1},\Xt{2};\Theta_0(\Xt{1}))$ above, we can decompose $C_n^{(2)}(\Xt{1},\Xt{2};\Theta_0(\Xt{1}))$ as follows:
\begin{align*}
&C_n^{(2)}(\Xt{1},\Xt{2};\Theta_0(\Xt{1})) \\
&= P_n\left[\varphi_C(\Xt{1},\Xt{2})+2A(g)b(\Xt{1};\theta^*)\right] + 2A(g)BP_n\left[\varphi_\theta(X;\theta^*)\right] - 2A(g)BP_n\left[b(\Xt{1};\hat\theta_n^{(-k)})\right] + o_P(n^{-1/2}) \\
&= P_n\left[\varphi_C(\Xt{1},\Xt{2})+2A(g)b(\Xt{1};\theta^*)\right] + 2A(g)BP_n\left[\varphi_\theta(X;\theta^*)-b(\Xt{1};\theta^*)\right] \\
&\quad - 2A(g)B^2P_n\left[\varphi_\theta(X;\theta^*)\right] + o_P(n^{-1/2}),
\end{align*}
where in the second step we use \eqref{eq:Pb2}.

Similar to $C_n^{(1)}(\Xt{1},\Xt{2};\Theta_0(\Xt{1}))$, we see that $C_n^{(2)}(\Xt{1},\Xt{2};\Theta_0(\Xt{1}))$ has the additional term, $-2A(g)B^2P_n\left[\varphi_\theta(X;\theta^*)\right]$, which is equal to the extra bias term in $C_n^{(1)}(\Xt{1},\Xt{2};\Theta_0(\Xt{1}))$ scaled by $-B$. Consider next the following, which attempts to debias $C_n^{(2)}(\Xt{1},\Xt{2};\Theta_0(\Xt{1}))$ by adding $2A(g)B^2P_n\left[b(\Xt{1};\hat\theta_n^{(-k)})\right]$:
\begin{align*}
&C_n^{(3)}(\Xt{1},\Xt{2};\Theta_0(\Xt{1})) \\
&= C_n^{(2)}(\Xt{1},\Xt{2};\Theta_0(\Xt{1})) + 2A(g)B^2P_n\left[b(\Xt{1};\hat\theta_n^{(-k)})\right] \\ 
&= P_n\left[\varphi_C(\Xt{1},\Xt{2})+2A(g)b(\Xt{1};\theta^*)\right] + 2A(g)BP_n\left[\varphi_\theta(X;\theta^*)-b(\Xt{1};\theta^*)\right] - 2A(g)B^2P_n\left[\varphi_\theta(X;\theta^*)\right] \\
&\quad+ 2A(g)B^2P_n\left[b(\Xt{1};\hat\theta_n^{(-k)})\right] + o_P(n^{-1/2}) \\
&= P_n\left[\varphi_C(\Xt{1},\Xt{2})+2A(g)b(\Xt{1};\theta^*)\right] + 2A(g)BP_n\left[\varphi_\theta(X;\theta^*)-b(\Xt{1};\theta^*)\right] - 2A(g)B^2P_n\left[\varphi_\theta(X;\theta^*)\right] \\
&\quad+ 2A(g)B^2\left(P_n\left[b(\Xt{1};\theta^*)\right] +  BP_n\left[\varphi_\theta(X;\theta^*)\right]\right) + o_P(n^{-1/2}) \\
&= P_n\left[\varphi_C(\Xt{1},\Xt{2})+2A(g)b(\Xt{1};\theta^*)\right] + 2A(g)BP_n\left[\varphi_\theta(X;\theta^*)-b(\Xt{1};\theta^*)\right]  \\
&\quad- 2A(g)B^2P_n\left[\varphi_\theta(X;\theta^*)-b(\Xt{1};\theta^*)\right]+ 2A(g)B^3P_n\left[\varphi_\theta(X;\theta^*)\right] + o_P(n^{-1/2}),
\end{align*}
where in the third equality we again rely on \eqref{eq:Pb2}.

In the above a clear pattern emerges: at each step, debiasing $P_n\left[\varphi_\theta(X;\theta^*)\right]$ with $P_n\left[b(\Xt{1};\hat\theta_n^{(-k)})\right]$ triggers the follow-on bias term $-BP_n\left[\varphi_\theta(X;\theta^*)\right]$. Repeating this process $M$ times yields
$$
C_n^{(M)}(\Xt{1},\Xt{2};\Theta_0(\Xt{1})) = C_n(\Xt{1},\Xt{2};\Theta_0(\Xt{1})) + 2A(g)\left(\sum_{m=0}^{M-1}(-B)^m\right)P_n\left[b(\Xt{1}_i;\hat\theta_n^{(-k)})\right],
$$
which expands as
\footnotesize
\begin{align*}
&C_n^{(M)}(\Xt{1},\Xt{2};\Theta_0(\Xt{1})) \\
&= C_n(\Xt{1},\Xt{2};\Theta_0(\Xt{1})) + 2A(g)\left(\sum_{m=0}^{M-1}(-B)^m\right)P_n\left[b(\Xt{1}_i;\hat\theta_n^{(-k)})\right] \\
&= P_n\left[\varphi_C(\Xt{1},\Xt{2})\right] + 2A(g)\left(\sum_{m=0}^{M-1} (-B)^m\right)\left(P_n\left[b(\Xt{1};\theta^*)\right] +  BP_n\left[\varphi_\theta(X;\theta^*)\right]\right) + o_P(n^{-1/2}) \\
&= P_n\left[\left(\Xt{1}+\Xt{2}-2\frac{N(\Xt{1})}{D(\Xt{1})}\right)g(\Xt{1})-2A(g)\varphi_\theta(X;\theta^*)+ 2A(g)\left(\sum_{m=0}^{M-1} (-B)^m\right)\left(b(\Xt{1};\theta^*) +  B\varphi_\theta(X;\theta^*)\right)\right] + o_P(n^{-1/2}) \\
&= P_n\left[\left(\Xt{1}+\Xt{2}-2\frac{N(\Xt{1})}{D(\Xt{1})}\right)g(\Xt{1})-2A(g)\left(
\left(I-\left(\sum_{m=0}^{M-1} (-B)^m  \right)B\right)\varphi_\theta(X;\theta^*)-\left(\sum_{m=0}^{M-1} (-B)^m\right)b(\Xt{1};\theta^*)\right)\right] + o_P(n^{-1/2}) \\
&= P_n\left[\left(\Xt{1}+\Xt{2}-2\frac{N(\Xt{1})}{D(\Xt{1})}\right)g(\Xt{1})-2A(g)\left(
\left(I+\sum_{m=1}^{M} (-B)^m  \right)\varphi_\theta(X;\theta^*)-\left(\sum_{m=0}^{M-1} (-B)^m\right)b(\Xt{1};\theta^*)\right)\right] + o_P(n^{-1/2}) \\
&= P_n\left[\left(\Xt{1}+\Xt{2}-2\frac{N(\Xt{1})}{D(\Xt{1})}\right)g(\Xt{1})-2A(g)\left(
\left(\sum_{m=0}^{M} (-B)^m  \right)\varphi_\theta(X;\theta^*)-\left(\sum_{m=0}^{M-1} (-B)^m\right)b(\Xt{1};\theta^*)\right)\right] + o_P(n^{-1/2}),
\end{align*}
\normalsize
where the second equality follows from the definition of $\varphi_C$ in \eqref{eq:varphi-C} as well as \eqref{eq:Pb2}, and the third equality from Corollary \ref{cor:delta}.

When $M\to\infty$, if $\max(|\lambda_{\max}(B)|, |\lambda_{\min}(B)|) \le 1 $, all geometric series in the above converge to $(I+B)^{-1}$, yielding $D_n(\xt{1},\xt{2};\Theta_0(\xt{1}))$.

\subsection{Bounding the eigenvalues of $B$}
\label{app:eigen}

Lemma \ref{lem:DnIF} requires that $I+B$ is invertible; equivalently, it requires that none of the eigenvalues of $B$, defined in \eqref{eq:B}, are equal to $-1$. Proposition \ref{prop:efficient} shows that this is guaranteed to be true when the estimator $\hat\theta_n$ with influence function $\varphi_\theta$ is efficient. We begin with Lemma \ref{lem:cov}, which provides a generic simplification of $B$. We state and prove both results for the independent and identically distributed case; similar arguments hold for the triangular array setting.

\begin{lemma}
\label{lem:cov}
In the setting of Lemma \ref{lem:DnIF}, $B=-\E\left[\varphi_\theta(X;\theta^*)s_{\Xt{1}}(\Xt{1};\theta^*)\right]$ where $s_{\Xt{1}}(\Xt{1};\theta^*)$ is the score function for $\Xt{1}$.
\end{lemma}

\begin{proof}
Recall that for any $\theta^*\in\Theta$, it holds that $\E_{\theta^*}[\varphi_\theta(X;\theta^*)]=0$. It follows that 
\begin{align*}
&\nabla_\theta\left(\E_{\theta}[\varphi_\theta(X;\theta)]\right)\vert_{\theta=\theta^*}=0 \\
\implies&\nabla_\theta\left(\int\varphi_\theta(x;\theta)f_X(x;\theta)dx\right)\vert_{\theta=\theta^*}=0 \\
\implies&\int\left(\nabla_\theta\left(\varphi_\theta(x;\theta)\right)\vert_{\theta=\theta^*}f_X(x;\theta^*)+\varphi_\theta(x;\theta^*)s_X(x;\theta^*)f_X(x;\theta^*)\right)dx=0 \\
\implies&\E_{\theta^*}[\nabla_\theta\left(\varphi_\theta(x;\theta)\right)\vert_{\theta=\theta^*}]+\E_{\theta^*}[\varphi_\theta(x;\theta^*)s_X(x;\theta^*)]=0.
\end{align*}

Next, observe that $B$ simplifies similarly as 
\begin{align*}
B&=\int \nabla_\theta\left(E_\theta[\varphi_\theta(X;\theta)|\Xt{1}=\xt{1}]\right)\Big\vert_{\theta=\theta^*}f_{\Xt{1}}(\xt{1};\theta^*)d\xt{1} \\
&=\int \nabla_\theta\left(\int\varphi_\theta(x;\theta)f_{X|\Xt{1}}(x;\theta)dx\right)\Big\vert_{\theta=\theta^*}f_{\Xt{1}}(\xt{1};\theta^*)d\xt{1} \\
&=\int \int\left(\nabla_\theta\left(\varphi_\theta(x;\theta)\right)\Big\vert_{\theta=\theta^*}f_{X|\Xt{1}}(x;\theta^*) + \varphi_\theta(x;\theta^*)s_{X|\Xt{1}}(x;\theta^*)f_{X|\Xt{1}}(x;\theta^*) \right)dxf_{\Xt{1}}(\xt{1};\theta^*)d\xt{1} \\
&=\int \int\nabla_\theta\left(\varphi_\theta(x;\theta)\right)\Big\vert_{\theta=\theta^*}f_{X,\Xt{1}}(x,\xt{1};\theta^*)dxd\xt{1} +\int\int \varphi_\theta(x;\theta^*)s_{X|\Xt{1}}(x;\theta^*)f_{X,\Xt{1}}(x,\xt{1};\theta^*)dxd\xt{1} \\
&=\E_{\theta^*}[\nabla_\theta\left(\varphi_\theta(X;\theta)\right)\vert_{\theta=\theta^*}] + \E_{\theta^*}[\varphi_\theta(X;\theta^*)s_{X|\Xt{1}}(X;\theta^*)] \\
&=\E_{\theta^*}[\nabla_\theta\left(\varphi_\theta(X;\theta)\right)\vert_{\theta=\theta^*}] + \E_{\theta^*}[\varphi_\theta(X;\theta^*)\{s_X(X;\theta^*)-s_{\Xt{1}}(\Xt{1};\theta^*)\}] \\
&=\E_{\theta^*}[\nabla_\theta\left(\varphi_\theta(X;\theta)\right)\vert_{\theta=\theta^*}] + \E_{\theta^*}[\varphi_\theta(X;\theta^*)s_X(X;\theta^*)] - \E_{\theta^*}[\varphi_\theta(X;\theta^*)s_{\Xt{1}}(\Xt{1};\theta^*)] \\
&= - \E_{\theta^*}[\varphi_\theta(X;\theta^*)s_{\Xt{1}}(\Xt{1};\theta^*)],
\end{align*}
where in the last line we used the identity from the previous string of statements.
\end{proof}

\begin{proposition}
\label{prop:efficient}
In the setting of Lemma \ref{lem:cov}, suppose further that $\hat\theta_n$ is an efficient estimator of $\theta^*$. Then the eigenvalues of B are all contained in $(-1,0)$.
\end{proposition}

\begin{proof}
Since $\hat\theta_n$ is efficient, $\varphi_\theta(x;\theta^*)=\mathcal{I}_X^{-1}(\theta^*)s_X(x;\theta^*)$ where $\mathcal{I}_X(\theta^*)$ is the information matrix. Then, using the expression in Lemma \ref{lem:cov},
\begin{align*}
B&=-\E_{\theta^*}[\varphi_\theta(X;\theta^*)s_{\Xt{1}}(\Xt{1};\theta^*)] \\
&=-\mathcal{I}_X^{-1}(\theta^*)\E_{\theta^*}[s_X(X;\theta^*)s_{\Xt{1}}(\Xt{1};\theta^*)] \\
&=-\mathcal{I}_X^{-1}(\theta^*)\E_{\theta^*}[\{s_{\Xt{1}}(\Xt{1};\theta^*)+s_{X|\Xt{1}}(X;\theta^*)\}s_{\Xt{1}}(\Xt{1};\theta^*)] \\
&=-\mathcal{I}_X^{-1}(\theta^*)\E_{\theta^*}[s_{\Xt{1}}(\Xt{1};\theta^*)s_{\Xt{1}}(\Xt{1};\theta^*)] \\
&=-\mathcal{I}_X^{-1}(\theta^*)\mathcal{I}_{\Xt{1}}(\theta^*),
\end{align*}
where in the second-to-last equality, $\E_{\theta^*}[s_{X|\Xt{1}}(X;\theta^*)s_{\Xt{1}}(\Xt{1};\theta^*)]=0$ by the orthogonality of $\Xt{1}$ and $X|\Xt{1}$, and the last equality follows from the definition of Fisher information.

The above shows that $B$ is the negative of the ratio of  $\mathcal{I}_X^{-1}(\theta^*)$ and $\mathcal{I}_{\Xt{1}}(\theta^*)$. This ratio is bounded between $0$ and $1$ under the assumption that the noise introduced to construct $\Xt{1}$ from $X$ is non-trivial (which can always be achieved since the distribution of $W$ is user-selected).
\end{proof}

\subsection{The effect of estimating $A(g)(I+B)^{-1}$}
\label{app:AB}

Recall that the characterization of $D_n(\xt{1},\xt{2};\Theta_0(\xt{1}))$, defined in \eqref{eq:Dn}, in Lemma \ref{lem:DnIF} depends on the constant $A(g)(I+B)^{-1}$ where $A(g)$ is defined in Corollary \ref{cor:delta} and $B$ is defined in \eqref{eq:B}. In practice, $A(g)(I+B)^{-1}$ is unknown and therefore must be estimated. Consider estimators $\hat A(g)$ and $\hat B$ that satisfy $\hat A(g) \overset{P}\to A(g)$ and $\hat B \overset{P}\to B$ so that $\hat A(g)(I+\hat B)^{-1}\overset{P}\to A(g)(I+B)^{-1}$. Here we show that in the setting of Theorem \ref{thm:Dn}, replacing $A(g)(I+B)^{-1}$ with $\hat A(g)(I+\hat B)^{-1}$ does not change the asymptotic distribution of $D_n(\xt{1},\xt{2};\Theta_0(\xt{1}))$. 

As in the previous section, let $b(\xt{1}_i;\theta)=\frac{\E_{U_i}[\varphi_\theta(U_i;\theta)f_{X}(U_i;\eta_i(\theta))]}{\E_{U_i}[f_{X}(U_i;\eta_i(\theta))]}$ where $U_i\sim R(\xt{1}_i,\Sigma)$, then consider 
$$
\tilde D_n(\xt{1},\xt{2};\Theta_0(\xt{1})) = C_n(\Xt{1},\Xt{2};\Theta_0(\Xt{1})) + 2\hat A(g)(I+\hat B)^{-1}P_n\left[b(\Xt{1};\hat\theta_n^{(-k)})\right].
$$
Observe that 
\begin{align*}
&\tilde D_n(\xt{1},\xt{2};\Theta_0(\xt{1})) \\
&=D_n(\Xt{1},\Xt{2};\Theta_0(\Xt{1})) + 2\left(\hat A(g)(I+\hat B)^{-1}-A(g)(I+B)^{-1}\right)P_n\left[b(\Xt{1};\hat\theta_n^{(-k)})\right] \\
&=D_n(\Xt{1},\Xt{2};\Theta_0(\Xt{1})) + 2\left(\hat A(g)(I+\hat B)^{-1}-A(g)(I+B)^{-1}\right)P_n\left[b(\Xt{1};\theta^*) +  B\varphi_\theta(X;\theta^*)\right] +o_P(n^{-1/2}),
\end{align*}
where the second equality follows from \eqref{eq:Pb2}.

Since $\hat A(g)(I+\hat B)^{-1}-A(g)(I+B)^{-1}=o_P(1)$ and $P_n\left[b(\Xt{1};\theta^*) +  B\varphi_\theta(X;\theta^*)\right]=O_p(n^{-1/2})$, it follows that 
$$
\tilde D_n(\xt{1},\xt{2};\Theta_0(\xt{1})) = D_n(\xt{1},\xt{2};\Theta_0(\xt{1})) + o_P(n^{-1/2}).
$$

In practice, we prefer to use the cross-fit plug-in estimators. That is, we define $\hat A_n^{(-k)}(g)$ and $\hat B_n^{(-k)}$ as Monte Carlo approximations of the integrals
$$
\int\left\{\E_{U_s}[U_s\nabla_\theta f_X(x;\eta_{i'}(\theta))^\top \mid_{\theta=\hat\theta_n^{(-k)}}] 
-\frac{N(s)}{D(s)}\E_{U_{s}}[\nabla_\theta f_X(x;\eta_{i'}(\theta))^\top \mid_{\theta=\hat\theta_n^{(-k)}}] \right\}\frac{g(s)}{D(s)} dP_{\xt{1};\hat\theta_n^{(-k)}}(s)
$$
and 
$$
\int\left(\nabla_\theta\frac{\E_{U_s}[\varphi_\theta(U_s;\theta)f_{X}(U_s;\eta_s(\theta))]}{\E_{U_s}[f_{X}(U_s;\eta_s(\theta))]}\Big{\vert}_{\theta=\hat\theta_n^{(-k)}}\right)^\top dP_{\xt{1};\hat\theta_n^{(-k)}}(s),
$$
respectively, where the notation $dP_{\xt{1};\hat\theta_n^{(-k)}}$ indicates that the integral is taken with respect to the measure implied by $\hat\theta_n^{(-k)}$.

\section{Additional details on the simulation studies}
\label{app:sims}

\subsection{Details for Section \ref{subsec:clustering}}
\label{app:cluster}

Algorithm \ref{alg:cluster} presents the details of conducting inference after clustering via orthogonalization.

\begin{algorithm}[Inference after clustering via orthogonalization]
\label{alg:cluster}
\textcolor{white}{.}

Input: Observed data $x_i$ drawn according to \eqref{eq:cluster}; a positive integer $c$; a clustering algorithm $\mathcal{M}$; a test function $g:\mathcal{X}^{(1)}\rightarrow\mathbb{R}$; and the number of cross-fitting folds $K$ (default: 5).  
\begin{enumerate}
    \item Construct $\xt{1}_i$ and $\xt{2}_i$ using Algorithm \ref{alg:split} with inputs $x_1,\dots,x_n$ and $R(\phi,\Sigma)\overset{D}{=} \text{DiscreteUniform}(\phi-c,\phi+c)$. 
    \item Apply the clustering algorithm $\mathcal{M}$ to $\xt{1}$ and extract two estimated clusters, $\widehat{\mathcal{C}}_1$ and $\widehat{\mathcal{C}}_2$. Define the selected hypothesis $H_0(\xt{1}):(\lambda_{ij},\pi_{ij})=(\lambda_{i'j},\pi_{i'j}),\quad\forall i,i'\in \widehat{\mathcal{C}}_1\cup\widehat{\mathcal{C}}_2,j=1,\dots,p$.
    \item Partition $\widehat{\mathcal{C}}_1\cup\widehat{\mathcal{C}}_2$ into $K$ folds, and let $\widehat{\mathcal{C}}^{(k)}$ denote the observations assigned to the $k$th fold.
    \item For $k=1,\dots,K$, let $\left(\hat\lambda_1^{(-k)},\dots,\hat\lambda_p^{(-k)},\hat\pi_1^{(-k)},\dots,\hat\pi_p^{(-k)}\right)$ denote the maximum likelihood estimate computed using the observations $X_i$ for $i\in \widehat{\mathcal{C}}_1\cup\widehat{\mathcal{C}}_2\setminus\widehat{\mathcal{C}}^{(k)}$; note that the likelihood implied by \eqref{eq:cluster} must be renormalized to account for the fact that the observations $X_i$ for $i\in \widehat{\mathcal{C}}_1\cup\widehat{\mathcal{C}}_2\setminus\widehat{\mathcal{C}}^{(k)}$ are restricted to a subset of $\mathcal{X}^{(1)}$. Then compute the conditional mean using \eqref{eq:hstar}; since the coordinates are independent and the data are discrete, all expectations can be computed exactly.
    \item Test the hypothesis $H_0'(\xt{1},g)$ defined in \eqref{eq:H0prime2} using $\vartheta_\alpha^ {\Theta_0(\xt{1})}\left(\Xt{1},\Xt{2}\right)$ defined in \eqref{eq:test2}. 
\end{enumerate}
\end{algorithm}

\subsection{Inference after changepoint detection simulation study}
\label{app:cp}

Given an ordered sequence $X_i$ for $i=1,\dots,n$, the goal of changepoint detection is to identify the indices at which some characteristic of the data-generating distribution changes. Many algorithms exist for detecting candidate changepoints; see \cite{truong2020selective} and \cite{fearnhead2020relating} for reviews. However, testing whether an identified changepoint represents a true change is challenging since estimated changepoints are data dependent; that is, the same data are used both to detect and test changepoints.

When the data are Gaussian, \cite{hyun2021post}, \cite{jewell2022testing}, and \cite{carrington2025improving} offer conditional selective inference procedures. Outside of the Gaussian setting, \cite{dharamshi2023generalized} apply data thinning to conduct valid inference after changepoint detection when the data belong to a suitable parametric family, and \cite{opss} propose order-preserved sample splitting in which odd-indexed observations are used to identify changepoints and even-indexed observations are used for inference. The latter approach has the disadvantage of rigidly assigning exactly half of the data for each step. Overall, the limited set of options outside of the Gaussian setting is problematic, as \cite{fearnhead2020relating} note that many algorithms overestimate the number of changepoints for non-Gaussian data.

In this section we use Algorithm \ref{alg:sel} to conduct inference after changepoint detection when $X_1,\dots,X_n$ are generated from a Student's $t$-distribution with $5$ degrees of freedom. Specifically, for $n\in\{200,500,1000\}$, we consider data generated according to 
\begin{equation}
\label{eq:cp}
X_i = \mu_i + \epsilon_i; \quad \epsilon_i \sim t_5,
\end{equation}
for $i=1,\dots,n$, where $\mu_i$ is the mean of observation $i$.

We begin with a ``null" setting in which $\mu_i=2$ for all $i=1,\dots,n$. Here $X_1,\dots,X_n$ are independent and identically distributed; there are thus no true changepoints in the sequence. For each setting of $n$, we draw $1,000$ replicates from \eqref{eq:cp}, then apply Algorithm \ref{alg:cp}, which is a variant of Algorithm \ref{alg:sel} specialized to the context of inference after changepoint detection. In particular, for each replicate, we apply Algorithm \ref{alg:cp} so that in Step 1, we decompose $x$ into $\xt{1}$ and $\xt{2}$ with Gaussian noise with variance $c\in\{0.5,1,2\}$; in Step 2, we use binary segmentation \citep{scott1974cluster} with a minimum segment length of 30 to estimate at most 4 changepoints; and in Step 5, we test $H_0(\xt{1})$ by testing the reformulated $H_0'(\xt{1},g)$ with two choices of $g$: (i) $g(\xt{1}_i)=I(i \in \widehat{\mathcal{S}}_{post})$ and (ii) $g(\xt{1}_i)=\xt{1}_iI(i\in\widehat{\mathcal{S}}_{pre}\cup\widehat{\mathcal{S}}_{post})$. Analogously to Section \ref{subsec:clustering}, we do not consider $g(\xt{1}_i)=I(i\in\widehat{\mathcal{S}}_{pre}\cup\widehat{\mathcal{S}}_{post})$ as it leads to degeneracy in $\varphi_D$ (see Remark \ref{rem:degenerate}).

\begin{algorithm}[Inference after changepoint detection via orthogonalization]
\label{alg:cp}
\textcolor{white}{.}

Input: Observed data $x_i$ drawn according to \eqref{eq:cp}; a positive real number $c$; a changepoint detection algorithm $\mathcal{M}$; a test function $g:\mathcal{X}^{(1)}\rightarrow\mathbb{R}$; the number of cross-fitting folds $K$ (default: 10); and the number of Monte Carlo replicates $B$ (default: 5,000).  
\begin{enumerate}
    \item Construct $\xt{1}_i$ and $\xt{2}_i$ using Algorithm \ref{alg:split} with inputs $x_1,\dots,x_n$ and $R(\phi,\Sigma)\overset{D}{=} N(0,c)$. 
    \item Apply the changepoint detection algorithm $\mathcal{M}$ to $\xt{1}$, and let $\hat\tau$ denote an estimated changepoint. Define the selected hypothesis $H_0(\xt{1}):\mu_i=\mu_{i'},\quad\forall i,i'\in \widehat{\mathcal{S}}_{pre}\cup\widehat{\mathcal{S}}_{post}$, where $\widehat{\mathcal{S}}_{pre}$ denotes the segment immediately preceding $\hat\tau$ and $\widehat{\mathcal{S}}_{post}$ denotes the segment immediately following $\hat\tau$.
    \item Partition $\widehat{\mathcal{S}}_{pre}\cup\widehat{\mathcal{S}}_{post}$ into $K$ folds, and let $\widehat{\mathcal{S}}^{(k)}$ denote the observations assigned to the $k$th fold.
    \item For $k=1,\dots,K$, let $\hat\mu^{(-k)}=\frac{1}{|\widehat{\mathcal{S}}_{pre}\cup\widehat{\mathcal{S}}_{post}\setminus \widehat{\mathcal{S}}^{(k)}|}\sum_{i\in\widehat{\mathcal{S}}_{pre}\cup\widehat{\mathcal{S}}_{post}\setminus \widehat{\mathcal{S}}^{(k)}}X_i$, then compute the conditional mean using \eqref{eq:hstar} with the Monte Carlo strategy outlined in Section \ref{subsubsec:ale} with $B$ replicates.
    \item Test the hypothesis $H_0'(\xt{1},g)$ defined in \eqref{eq:H0prime2} using $\vartheta_\alpha^ {\Theta_0(\xt{1})}\left(\Xt{1},\Xt{2}\right)$ defined in \eqref{eq:test2}. 
\end{enumerate}
\end{algorithm}

Figure \ref{fig:cp-t1e} displays the p-values for this experiment. Each panel corresponds to a value of $n$ and displays the empirical quantiles of the p-values against the quantiles of the uniform distribution; the choice of $g$ is indicated by colour and the choice of $c$ is indicated by line type. The Type I error rate is controlled in all settings.

\begin{figure}[h]
\centering
\includegraphics[width=0.9\linewidth]{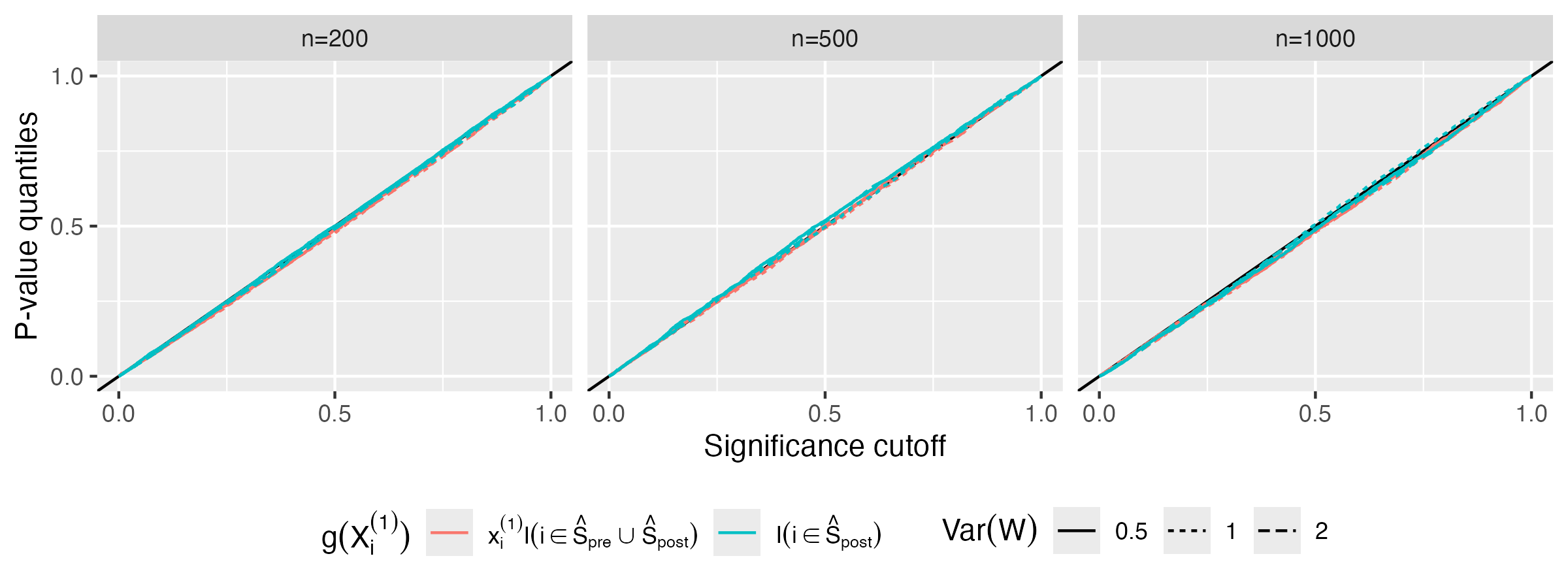}
\caption{Type I error results for the ``null" setting of the simulation described in Supplement \ref{app:cp}. Each panel displays a QQ-plot of the empirical quantiles of the observed p-values against the quantiles of a $\text{Uniform}(0,1)$ distribution. The Type I error rate is controlled on average, across realizations of $\Xt{1}$, for all settings of $n$ (indicated by panel), choices of test function $g$ (indicated by colour), and choices of noise variance $c$ (indicated by line type).} 
\label{fig:cp-t1e}
\end{figure}

Next, to assess the power of our approach, we consider an ``alternative" setting in which for all $i=1,\dots,n$, $\mu_i=2+d\cdot (2I(i \le n/2)-1)$ for $d\in\{0.25,0.5,\dots,2.5\}$. For each combination of $n$ and $d$, we draw 500 replicates of $X$ in \eqref{eq:cp}, then apply Algorithm \ref{alg:cp} as in the ``null" setting. Figure \ref{fig:cp-pow} displays the power of Algorithm \ref{alg:cp} as a function of the effect size, which is defined as $\Delta=\left|\frac{1}{|\widehat{\mathcal{S}}_{pre}|}\sum_{i\in\widehat{\mathcal{S}}_{pre}}\mu_i-\frac{1}{|\widehat{\mathcal{S}}_{post}|}\sum_{i\in\widehat{\mathcal{S}}_{post}}\mu_i\right|$. Each panel corresponds to a value of $n$ and displays a (smoothed) power curve for each choice of $g$ (indicated by colour) and choice of $c$ (indicated by line type). As in the clustering simulation in Section \ref{subsec:clustering}, power increases as a function of $n$ and $\Delta$. Power also increases with $c$; this is expected: as the amount of noise added to the selection fold $\Xt{1}$ increases, more information is reserved for inference \citep{neufeld2023data}. In contrast with Section \ref{subsec:clustering}, choosing $g$ to be an indicator for $\widehat{\mathcal{S}}_{post}$ has higher power than the value of $\xt{1}_i$ itself, suggesting that the optimal choice of $g$ is context-specific.

\begin{figure}[h]
\centering
\includegraphics[width=0.9\linewidth]{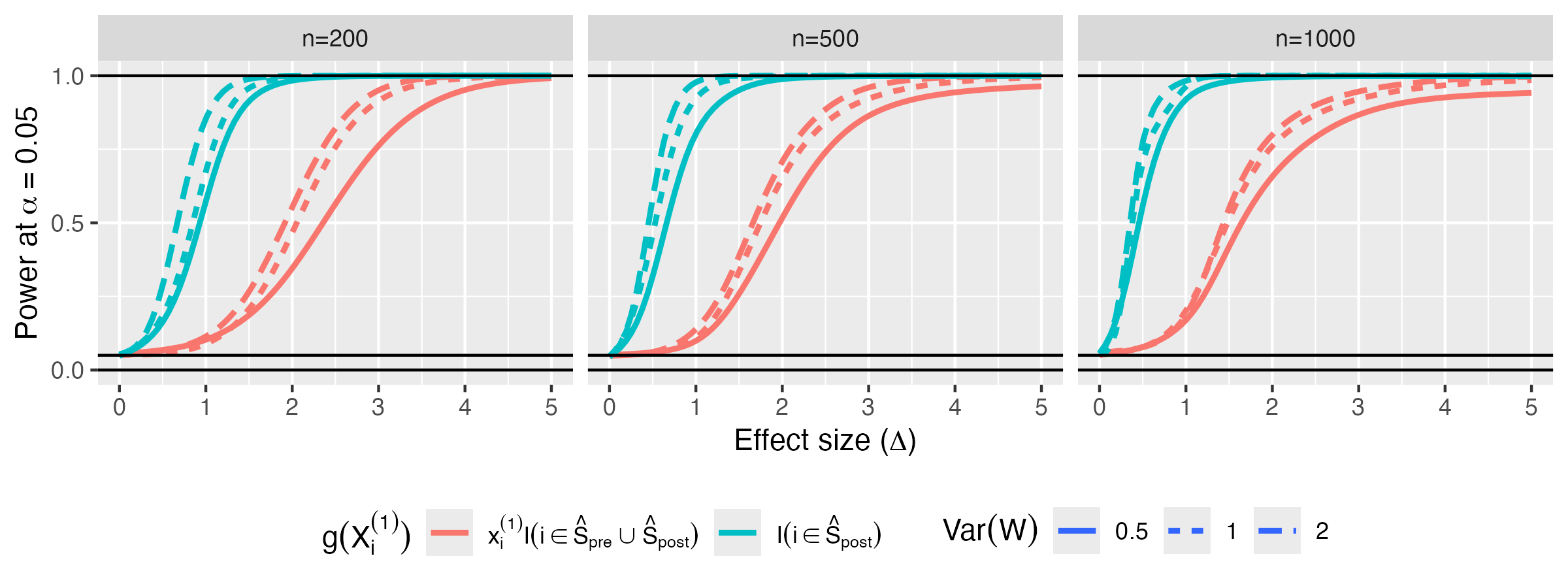}
\caption{Smoothed power results for the ``alternative" setting of the simulation described in Supplement \ref{app:cp}. Each panel displays six power curves as a function of the effect size $\Delta$ (defined in Supplement \ref{app:cp}), one for each combination of $g$ (indicated by colour) and $c$ (indicated by line type). Each panel corresponds to a setting of $n$. In all cases, power increases as a function of $c$, $n$, and $\Delta$; power is also greater for $g(\xt{1}_i)=I(i \in \widehat{\mathcal{S}}_{post})$ as compared to $g(\xt{1}_i)=\xt{1}_iI(i\in\widehat{\mathcal{S}}_{pre}\cup\widehat{\mathcal{S}}_{post})$.}
\label{fig:cp-pow}
\end{figure}

\end{document}